\documentclass[12pt]{article}
\usepackage{amsmath,amssymb,amsthm,amsfonts,bm,latexsym,graphicx,float, caption, subcaption,yfonts,cite}
\graphicspath{{figures/}}
\usepackage{setspace}
\usepackage[export]{adjustbox}
\def\R{\Bbb{R}}
\def\C{\mathfrak{C}}
\usepackage[font=small,format=plain,labelfont=bf,textfont={footnotesize,it},justification=justified,singlelinecheck=true]{caption}

\newtheorem{theorem}{Theorem}
\newtheorem{prop}[theorem]{Proposition}

\newtheorem{remark}[theorem]{Remark}
\newtheorem{corollary}[theorem]{Corollary}
\newtheorem{definition}[theorem]{Definition}

\everymath{\displaystyle}

  \title{3D Image Reconstruction from Compton Camera Data}
\author{Peter Kuchment\thanks{Department of Mathematics, Texas A$\&$M University, College Station, TX 77843-3368, USA, e-mail: kuchment@math.tamu.edu} and Fatma Terzioglu\thanks{Same department, e-mail: fatma@math.tamu.edu}}
\date{}
\begin{document}
\maketitle
\begin{abstract}
In this paper, we address analytically and numerically the inversion of the integral transform (\emph{cone} or \emph{Compton} transform) that maps a function on $\R^3$ to its integrals over conical surfaces. It arises in a variety of imaging techniques, e.g. in astronomy, optical imaging, and homeland security imaging, especially when the so called Compton cameras are involved.

Several inversion formulas are developed and implemented numerically in $3D$ (the much simpler $2D$ case was considered in a previous publication). An admissibility condition on detectors geometry is formulated, under which all these inversion techniques will work.

\end{abstract}
\section*{Introduction}
In this paper, we address analytic and numerical aspects of inversion of the integral transform that maps a function on $\R^n$ to its integrals over conical surfaces (with main concentration on the $3D$ case, while the much simpler $2D$ case was treated in \cite{Terzioglu}). It arises in a variety of imaging techniques, e.g. in optical imaging \cite{Florescu}, but most prominently when the so called \emph{Compton cameras} are used, e.g. in astronomy, SPECT medical imaging \cite{Todd,Singh}, as well as in homeland security imaging \cite{ADHKK,ACCHKOR,XMCK,Hristova}. We will call it \emph{cone} or \emph{Compton} transform (in 2D, the names \emph{V-line transform} and \emph{broken ray transform} are also used).

Being already used in astronomy, the application of Compton cameras in nuclear medicine was first proposed in \cite{Todd} as an alternative to gamma (or Anger) cameras used in medical SPECT (Single Photon Emission Tomography) imaging. The drawback of conventional gamma cameras is that they utilize mechanical collimation in order to determine the direction of an incoming gamma photon. The signal acquired by a gamma camera is weak because only the gamma-rays approaching the detector in a very small angle of directions (see Fig. \ref{fig:collimation&scatter}(a)) can pass through the collimator \cite{Basko}. In addition, the camera must be rotated to obtain projections from different directions.

On the other hand, Compton cameras make use of the Compton scattering effect (see Fig. \ref{fig:collimation&scatter}(b)) to locate the radioactive source. The absence of mechanical collimation resolves the issue of low efficiency and the need for rotating the camera.  It also facilitates the design of hand-held devices \cite{Kishimoto}.
\begin{figure}[H]
\begin{center}
        \begin{subfigure}[b]{0.4\textwidth}
                \includegraphics[width=\textwidth]{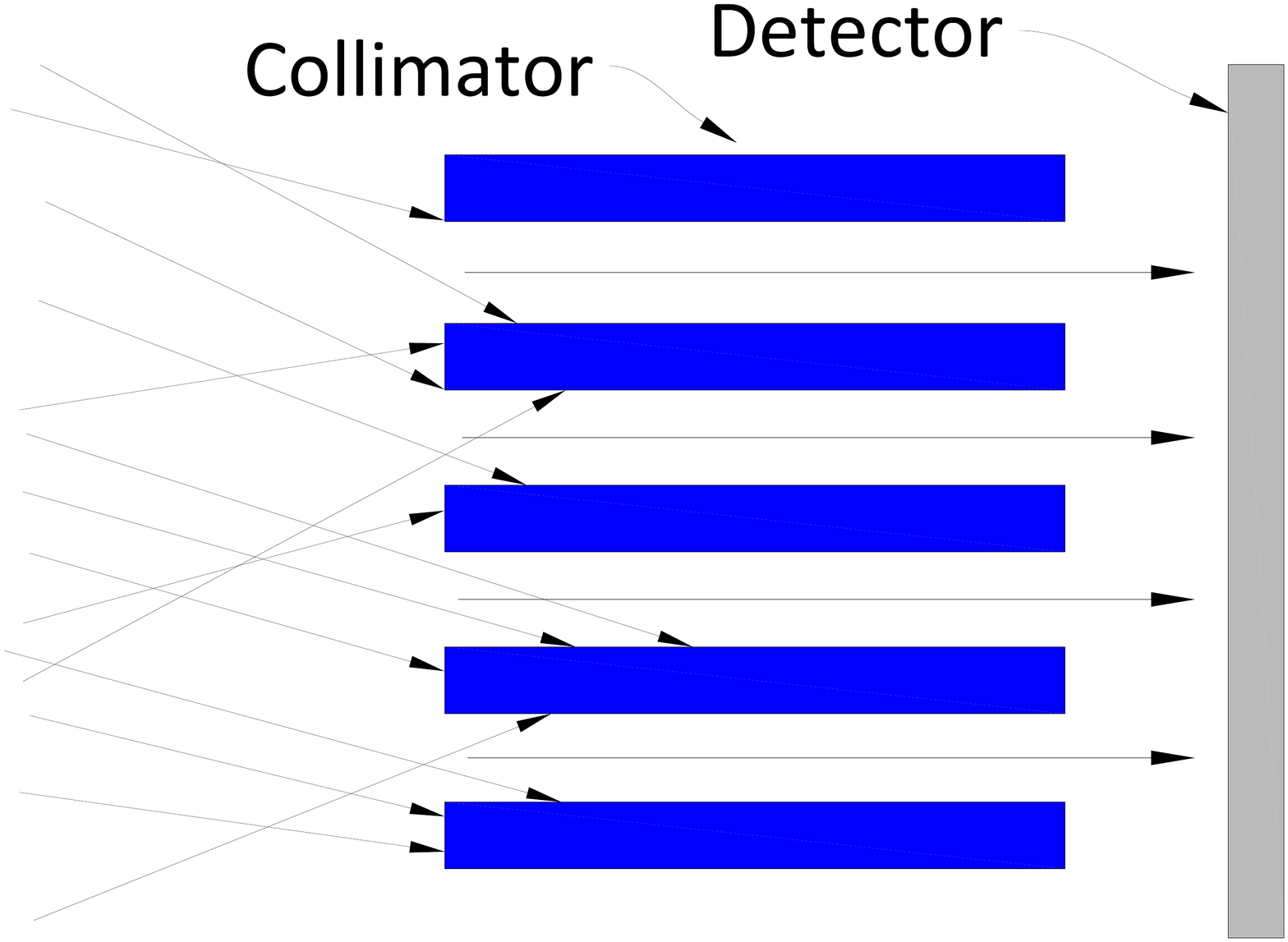}
                \caption{}
        \end{subfigure}
        \hspace{2em}
        \begin{subfigure}[b]{0.4\textwidth}
                \includegraphics[width=\textwidth]{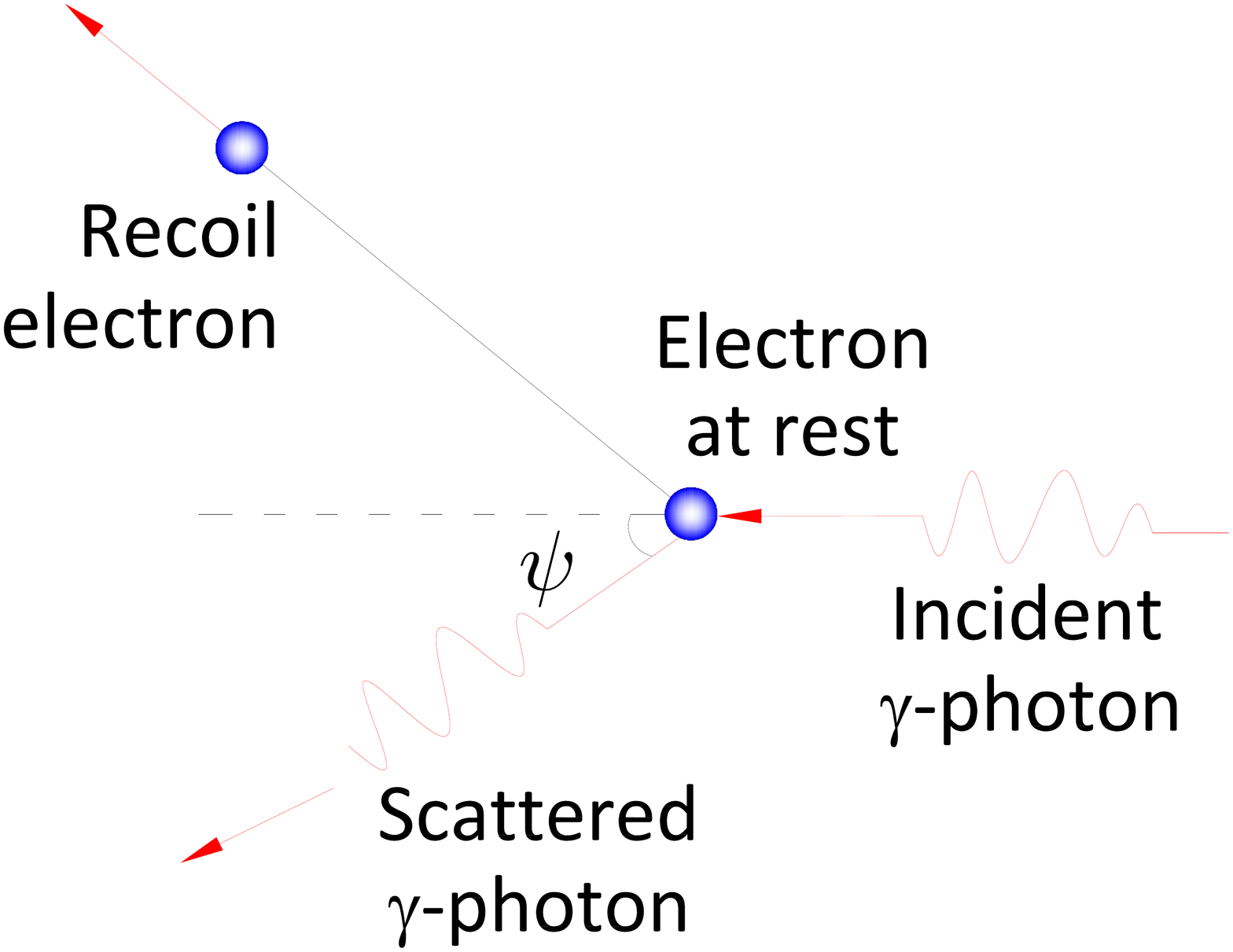}
                \caption{}
        \end{subfigure}
        \caption{Principles of  mechanical collimation (a) and Compton scattering (b).}
        \label{fig:collimation&scatter}
\end{center}
\end{figure}

A Compton camera consists of two parallel position and energy sensitive detectors (see Fig. \ref{fig:camera&2Dcone}(a)). When an incoming gamma photon hits the camera, it undergoes Compton scattering in the first detector (scatterer) and photoelectric absorption in the second detector (absorber). In both interactions, the positions $u$ and $v$ and the energies $E_1$ and $E_2$ of the photon are recorded. The scattering angle $\psi$ and a unit vector $\beta$ are calculated from the data as follows (see e.g. \cite{Todd}):
\begin{equation}
\cos\psi=1-\frac{mc^2E_1}{(E_1+E_2)E_2} \quad \quad \quad \quad \beta=\frac{u-v}{|u-v|}.
\end{equation}
Here, $m$ is the mass of the electron and $c$ is the speed of light.

\begin{figure}[H]
\begin{center}
        \begin{subfigure}[b]{0.45\textwidth}
                \includegraphics[width=\textwidth]{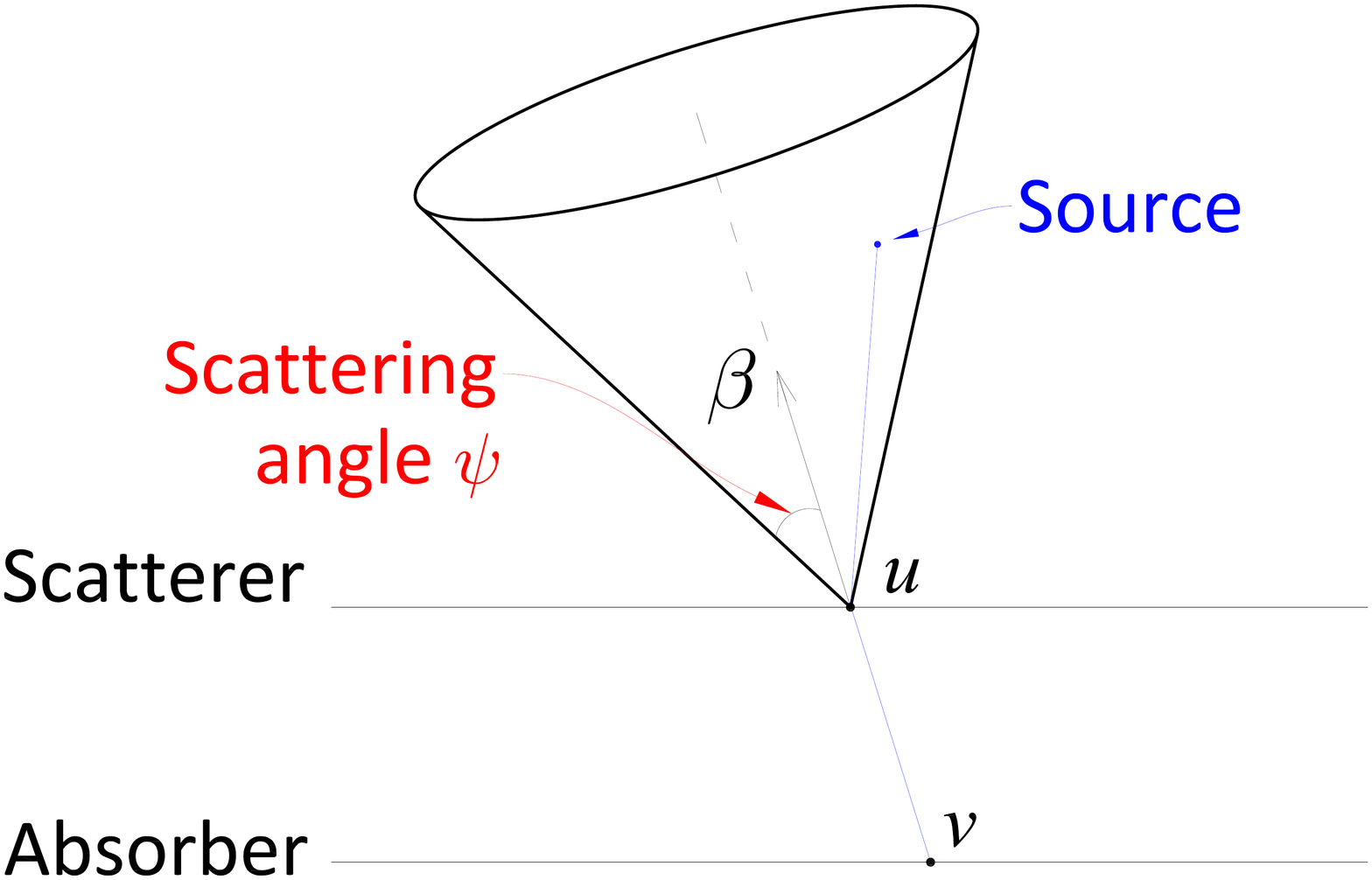}
                \caption{}
        \end{subfigure}
        \begin{subfigure}[b]{0.45\textwidth}
        \includegraphics[width=\textwidth]{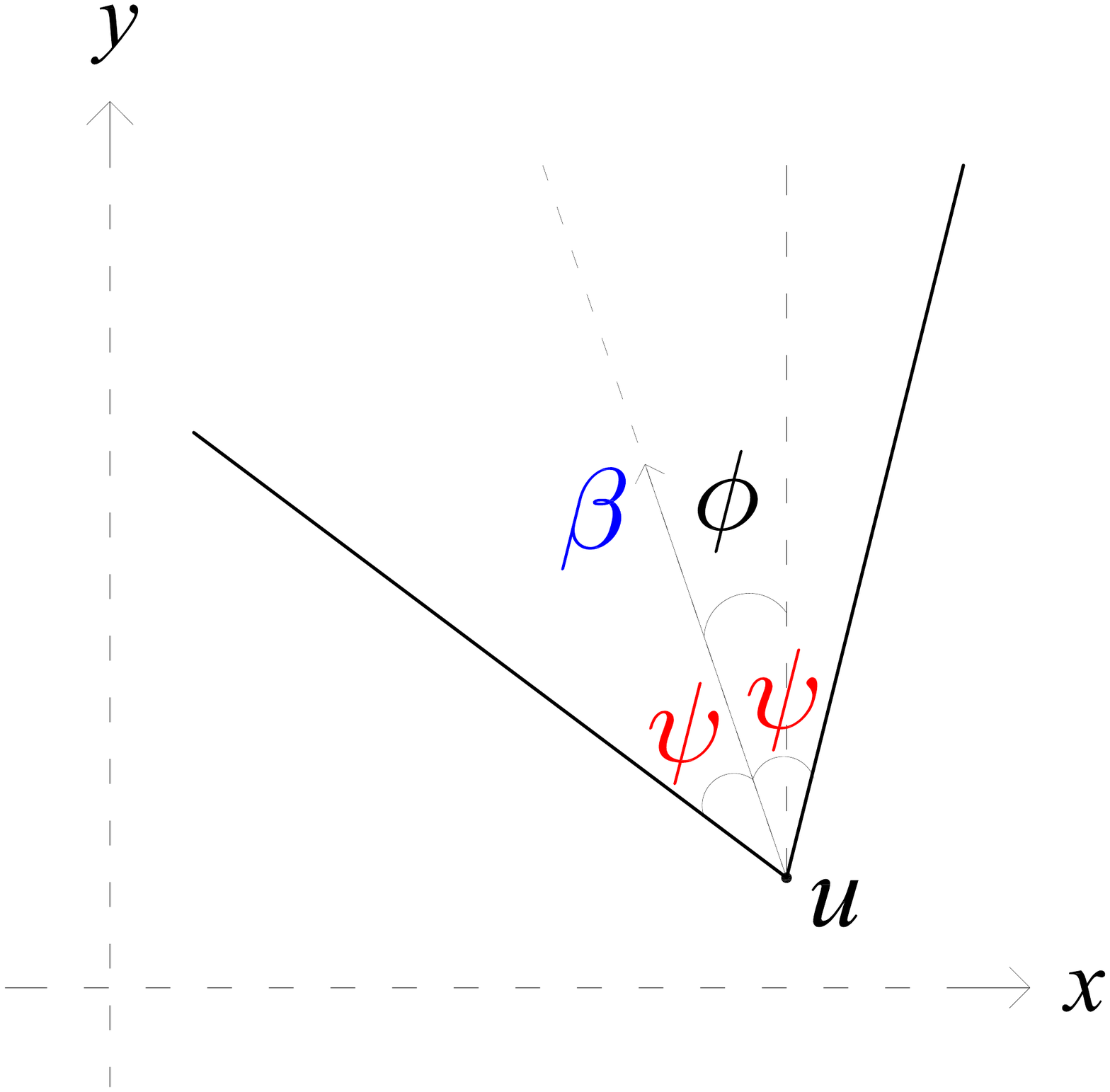}
        \caption{}
       \end{subfigure}
        \caption{Schematic representation of a Compton camera (a) and a cone in 2D (b).}
        \label{fig:camera&2Dcone}
\end{center}
\end{figure}
From the knowledge of the scattering angle $\psi$ and the vector $\beta$, one can conclude that the photon originated from the surface of the cone with central axis $\beta$, vertex $u$ and opening angle $\psi$ (see Fig. \ref{fig:camera&2Dcone}(a)). One can argue that the data provided by Compton camera are integrals of the distribution of the radiation sources over conical surfaces having vertex on the scattering detector\footnote{It has been mentioned in various papers, e.g. \cite{Basko,Smith,Maxim} that, depending upon the engineering of the detector, various weights can appear in the surface integral. Here we concentrate on the case of pure surface measure on the cone. As it is written in \cite{Basko}, ``It is not clear at this time if either
of these ... models accurately represents projections of a Compton camera. More work
needs to be done to determine the validity of the assumptions''.}. The operator that maps source intensity distribution function $f(x)$ to its integrals over these cones is called the \emph{cone} or \emph{Compton transform}. The goal of Compton camera imaging is to recover source distribution from this data \cite{ADHKK}.

In the Compton camera imaging applications mentioned above, the vertex of the cone is located on the detector array, while in some other applications the vertices are not restricted, although some other conditions are imposed on the cones (e.g., fixed axis directions). Also, in the Compton case, the data from all cones emanating from a given detector position is collected, while in some other applications only some cones (e.g., those with a prescribed axial direction) with a given vertex are involved.  Having the Compton imaging in mind, we thus follow the started in \cite{Terzioglu} line of studying analytic and numerical properties of the general cone transform, where all cones with a given vertex are accounted for, with the hope of obtaining consequences for more restricted version arising in practice (e.g., in Compton camera imaging). This partially  materialized in \cite{Terzioglu} in the much simpler $2D$ case. Here we address the $n$-dimensional situation (with main emphasis on $n=3$) and implement numerically some inversion formulas from  \cite{Terzioglu}, as well as some new ones developed below.

The geometry of the pair of the detectors of a Compton camera does not have to be planar. One can use, for instance, a curved pair of the scattering and absorbing detectors (see, e.g. \cite{Smith11}). In fact, for the reconstruction algorithms we develop, the geometry of detectors is irrelevant, as long as it satisfies the generous Admissibility Condition \ref{D:admis} in section 2. If this condition is violated, one can still use the algorithms, but then familiar limited data blurring artifacts \cite{Natt_old,KuchCBMS} will appear.

The problem of inverting the cone transform is over-determined (the space of cones in 3D with vertices on a detector surface is five-dimensional, three-dimensional in 2D). Without the restriction on the vertex, the dimensions are correspondingly six and four. One thus could restrict the set of cones, in order to get a non-over-determined problem  (e.g. \cite[and references therein]{Hristova,Hristova2015,MoonJung,Amb2,Amb3,Basko,Smith,Cree,Gouia-Zarrad-Ambarts,Haltmeier,Moon,Truong,NgTr2005}). In most of these considerations only a subset of cones with vertex at a given scattering detector is used. This means that most of the information already collected by the Compton camera is discarded. However, when the signals are weak (e.g. in homeland security applications \cite{ADHKK}), restricting the data would lead to essential elimination of the signal. We thus intend to use the data coming from all cones with vertices on the scattering detector. We also discuss viable restrictions on detector arrays.

In order to avoid being distracted from the main purpose of this text, we make in all theorems a severe overkill assumption that the functions in question belong to the Schwartz space $\mathcal{S}$ of smooth fast decaying functions. This allows us to skip discussions of applicability of various transforms. However, as it is in the case of Radon transform (see, e.g. \cite{Natt_old,Rubin}), the results have a much wider area of applicability, as in particular our numerical implementations show. The issues of appropriate functional spaces will be addressed elsewhere.

The paper is organized as follows. In the next section \ref{S:def}, we recall briefly some relevant transforms and their properties. In section \ref{S:inverse}, we obtain several procedures that convert the cone data to the Radon data of the same function, and thus allow for recovery of the function itself by using the well known filtered backprojection Radon transform inversion formulas. Section \ref{S:numerics} contains the results of numerical implementation of these approaches in $3D$. Remarks and conclusions can be found in section \ref{S:remarks}. The last section contains acknowledgments.

\section{Definitions}\label{S:def}

The surface of a circular cone in $\mathbb{R}^n$ can be parametrized by a tuple $(u, \beta, \psi)$, where $u \in \mathbb{R}^n$ is the cone's vertex, the unit vector $\beta \in S^{n-1}$ is directed along the cone's central axis, and the opening angle is $\psi \in (0,\pi)$ (see Fig. \ref{fig:camera&2Dcone}(a)). A point $x \in \mathbb{R}^n$ lies on the cone iff
\begin{equation}\label{cone eqn}
 (x-u)\cdot\beta=|x-u|\cos \psi.
\end{equation}
\begin{definition}
\normalfont The \emph{cone transform}  $C$ maps a function $f$ to its integrals over all circular cones in $\mathbb{R}^n:$
\begin{equation}
\label{cone trans}
 Cf(u,\beta,\psi):=\int\limits_{(x-u)\cdot\beta=|x-u|\cos \psi}f(x)dx,
\end{equation}
where $dx$ is the surface measure on the cone.
\end{definition}
In two dimensions, the equation \eqref{cone eqn} describes two rays with a common vertex (see Fig. \ref{fig:camera&2Dcone}(b)), which are also called as V-lines or broken lines in the literature. Then, the 2D cone transform of a function is its integral over these V-lines.
That is, for $\beta=\beta(\phi)=(\sin\phi, \cos \phi) \in S^1$, the 2D cone transform of a function $f \in \mathcal{S}(\mathbb{R}^2)$ is given by
\begin{equation}\label{2D_cone}
\begin{array}{l}
Cf(u, \beta(\phi), \psi)\\
=\int\limits_0^\infty [f(u+r(\sin(\psi+\phi),\cos(\psi+\phi)))+f(u+r(-\sin(\psi-\phi),\cos(\psi-\phi)))]dr.
\end{array}
\end{equation}

We also recall that the $n$-dimensional \emph{Radon transform} $R$ maps a function $f$ on $\mathbb{R}^n$ into the set of its integrals over the affine hyperplanes in $\mathbb{R}^n$. Namely, if $\omega \in S^{n-1}$ and $s \in \mathbb{R}$,
 \begin{equation}
\label{Def of Radon}
Rf(\omega, s)= \int\limits_{x \cdot \omega=s}f(x)dx.
\end{equation}
In this setting, the Radon transform of $f$ is the integral of $f$ over the hyperplane orthogonal to $\omega$ at the signed distance $s$ from the origin.

A variety of inversion formulas for the Radon transform are known (see, e.g. \cite{Natt_old,Helgason,Rubin,KuchCBMS}). We will only need the following formula (e.g., \cite{Natt_old}):
\begin{equation}\label{inverse_radon}
 f=\frac{1}{2}(2\pi)^{1-n}I^{-\alpha}R^\#I^{\alpha-n+1}Rf, \quad \quad \alpha < n.
\end{equation}
Here, $R^\#$ is the \emph{backprojection operator}  \cite{Natt_old} and $I^\alpha$, $\alpha<n$, is the \emph{Riesz potential} acting on a function $f$ as
$$\widehat{(I^\alpha f)}(\xi)=|\xi|^{-\alpha}\hat{f}(\xi),$$
where $\hat{f}$ is the Fourier transform of $f$ (see e.g. \cite{Helgason, KuchCBMS, Natt_old}).

The \emph{cosine transform} of a function $f \in C(S^{n-1})$ is defined by
\begin{equation}
 \C f(\omega)=\frac{1}{|S^{n-1}|}\int\limits_{S^{n-1}}f(\sigma)|\sigma \cdot \omega|d\sigma,
\end{equation}
for all $\omega \in S^{n-1}$ (see e.g. \cite{Gardner, Rubin}).

We will also need to use the \emph{Funk transform} (e.g., \cite{Rubin,GGG,Helgason,Palam}) that integrates a function on the sphere over all great circles (hyperplane sections). Several inversion formulas for the Funk transform exist in the literature \cite{Funk1,Helgason,GGG,Rubin,Palam}.

\section{Various inversion formulas for the cone transform}\label{S:inverse}
We start with a basic relationship between the cone, Radon and cosine transforms:
\begin{theorem}[\cite{Terzioglu}]
Let $f \in \mathcal{S}(\mathbb{R}^n)$ and $T_a$ be the translation operator in $\mathbb{R}^n$, defined as $T_af(x)=f(x+a)$ for $a \in \mathbb{R}^n$. Then, for any $u \in \mathbb{R}^n$ and $\beta \in S^{n-1}$, we have
 \begin{align}\label{int_rel2}
   \frac{1}{\pi}\int\limits_0^\pi Cf(u,\beta,\psi)\sin(\psi)d\psi= \frac{1}{|S^{n-1}|}\int\limits_{S^{n-1}}Rf(\omega, u\cdot \omega)|\omega \cdot \beta| d\omega  = \C (R(T_uf))(\beta),
   \end{align}
  where $|S^{n-1}|$ denotes the area of the sphere $S^{n-1}$.
\end{theorem}
The proof of this relation can be found in \cite{Terzioglu}. One can find a somewhat similar (albeit different, involving singular integration) formula in $3D$ in \cite{Smith}.

Since the cosine transform is a continuous automorphism of $C_{even}^\infty (S^{n-1})$ (see e.g. \cite{Gardner, Rubin}), and for any $f \in \mathcal{S}(\mathbb{R}^n)$, $Rf(\omega, 0)$ is an even function in $C^\infty(S^{n-1})$, we can recover the function $R(T_uf)$ by inverting the cosine transform.
Using the inversion formula for the cosine transform given in \cite[Chapter 5, Theorem 5.35]{Rubin}, we obtain the formulas given in Theorem \ref{Radon in terms of cone} below that recover the Radon data from the cone data. Then, inverting the Radon transform \cite{Natt_old}, one recovers the function $f$.

\begin{theorem}[\cite{Terzioglu}] \label{Radon in terms of cone}
Let $f \in \mathcal{S}(\mathbb{R}^n)$. For any $u \in \mathbb{R}^n$ and $\omega \in S^{n-1}$,
\begin{enumerate}
 \item if $n$ is odd,
  \begin{align}\label{Radon by cone_odd}
 Rf&(\omega, \omega \cdot u) =  \frac{\Gamma(\frac{n+1}{2})}{2\pi^{(n+1)/2}}\int\limits_{S^{n-1}}\int\limits_0^\pi Cf(u,\beta,\psi)\sin\psi d\psi d\beta\nonumber \\
 &-\frac{2\pi^{-n/2}}{\Gamma(\frac{n}{2})}P_{(n+1)/2}(\Delta_S)\left\{\int\limits_{S^{n-1}} \int\limits_0^\pi Cf(u,\beta,\psi)\log{\frac{1}{|\omega\cdot \beta|}}\sin\psi d\psi d\beta \right\},
 \end{align}

 \item if $n$ is even,
  \begin{equation}\label{Radon by cone_even}
  Rf(\omega, \omega \cdot u)=\frac{-2^{n-1}}{\Gamma(n-1)}\int\limits_0^\pi P_{n/2}(\Delta_S)F(Cf)(u,\omega,\psi)\sin\psi d\psi,
  \end{equation}
\end{enumerate}
where $F$ is the Funk transform, $\Delta_S$ is the Laplace-Beltrami operator on $S^{n-1}$ acting on $\omega$, and
$$P_r(\Delta_S)=4^{-r}\prod_{k=0}^{r-1}\left[-\Delta_S+(2k-1)(n-1-2k)\right].$$
\end{theorem}

This result, in particular, answers the question of what geometries of Compton detectors are sufficient for (stable) reconstruction of the function $f$. Indeed, formulas \eqref{Radon by cone_odd} and \eqref{Radon by cone_even} show that it is sufficient to have for any $\omega \in S^{n-1}$ and $s \in \mathbb{R}$ a detector location $u$ such that $\omega \cdot u = s$. This can be rephrased in a nice geometric way:
\begin{definition}[\emph{Compton Admissibility Condition}]\label{D:admis}
\normalfont We will call an array of Compton detectors \emph{admissible} (for a given region of space), if any hyperplane intersecting this region, intersects a detection site of the scattering detector.
\end{definition}

So, if a set $U$ of detectors is admissible for a region $D \in \mathbb{R}^n$, then the formulas \eqref{Radon by cone_odd} and \eqref{Radon by cone_even} enable one to reconstruct the Radon transform of any function $f$ supported inside $D$, and thus $f$ itself.

Here is a useful example of an application of the admissibility:
\begin{prop} \label{AdmissibilitySphere}
Suppose that $n=3$ and the detectors are placed on a sphere $S_r$ of radius $r$. We assume that the region for placing the object to be imaged is the concentric sphere $S_{r'}$ of radius $r'=r-\delta$ for some $\delta>0$. Then, any curve $U$ on $S_r$ that satisfies the condition below is admissible:
\begin{center}
Any circle on $S_r$ of radius $\rho \geq \sqrt{\delta(2r-\delta)}$ intersects $U$.
\end{center}
\end{prop}
\begin{proof} Indeed, every plane intersecting the interior of the sphere $S_{r'}$ intersect $S_r$ over a circle of radius $\rho \geq \sqrt{\delta(2r-\delta)}$ and thus contains at least one detector. \end{proof}
\begin{remark}\indent
\normalfont
\begin{enumerate}
\item The experience of Radon transform shows that uniqueness of reconstruction should hold for some non-admissible sets of detectors as well, although some (``invisible") sharp details will get blurred in the reconstruction (see, e.g. \cite[Ch. 7]{KuchCBMS}). The corresponding microlocal analysis of this issue will be done elsewhere.

\item The admissibility condition is not the minimal one. For instance, in the situation of Proposition \ref{AdmissibilitySphere}, the set of Compton data will still be 4-dimensional, and thus somewhat overdetermined. To avoid overdetermined data, one could use a single detection site, which would lead to some sharp features of the image being blurred.

\item In the cases of low signal-to-noise ratio (e.g. SPECT and especially homeland security imaging), one would prefer to use larger admissible sets of detectors (e.g. 2D rather than 1D arrays considered in Proposition \ref{AdmissibilitySphere}), which would allow introducing additional (weighted, if needed) averaging, in order to reduce the effects of the noise.

\item As it has been mentioned before, for all the reconstruction algorithms we develop in this text, the geometry of detectors is irrelevant, as long as it satisfies the generous Admissibility Condition \ref{D:admis} in section 2. If this condition is violated, one can still use the algorithms, but then familiar limited data blurring artifacts \cite{Natt_old,KuchCBMS} will appear.

\end{enumerate}
\end{remark}

A different approach to recovery of the Radon data from the Compton data comes from the following known relation (see \cite{Goodey}) between the cosine and Funk transforms:
\begin{equation}\label{E:cosinefunk}
(\Delta_S + n-1) \C  = F,
\end{equation}
where $\Delta_S$ is the Laplace-Beltrami operator on the sphere.

Indeed, applying $(\Delta_S + n-1)$ to \eqref{int_rel2}, we obtain
\begin{align}\label{funk of radon}
\Phi(u,\beta):=F(R(T_uf))(\beta) = \frac{(\Delta_S+n-1) }{\pi}\int\limits_0^\pi Cf(u,\beta,\psi)\sin\psi d\psi,
\end{align}
where  $\Delta_S$ acts in variable $\beta$.

We now use the inversion formula for the Funk transform given in \cite[Chapter 5, Theorem 5.37]{Rubin}, whose application to \eqref{funk of radon} leads to the following result.

\begin{theorem}\label{C2RbyFunk}
Let $f \in \mathcal{S}(\mathbb{R}^n)$. For any $u \in \mathbb{R}^n$ and $\omega \in S^{n-1}$,
\begin{equation}\label{E:Funk_inv}
\begin{array}{l}
 Rf(\omega, \omega \cdot u)=\frac{2^{n-1}}{(n-2)!} Q(\Delta_S) \left\{ \int\limits_{S^{n-1}} \Phi(u,\beta)\log{\frac{1}{|\omega\cdot \beta|}} d\beta \right\}\\
 + \frac{\Gamma(n/2)}{2\pi^{n/2}} \int\limits_{S^{n-1}} \Phi(u,\beta) d\beta,
\end{array}
\end{equation}
where $Q(\Delta_S) = 4^{(1-n)/2}\prod_{k=0}^{(n-3)/2}\left[-\Delta_S+(2k+1)(n-3-2k)\right].$
\end{theorem}
In particulary, in $3D$ one arrives to
\begin{corollary} For any $ u \in \mathbb{R}^3$ and $\omega \in S^2$,
\begin{equation}\label{inversion_funk}
\begin{array}{l}
 Rf(\omega, \omega \cdot u)=\frac{-\Delta_S}{2\pi} \left\{ \int\limits_{S^{n-1}} \Phi(u,\beta)\log{\frac{1}{|\omega\cdot \beta|}} d\beta \right\}\\
 + \frac{1}{4\pi} \int\limits_{S^{n-1}} \Phi(u,\beta) d\beta.
 \end{array}
\end{equation}
\end{corollary}

\section{Reconstructions in 3D}\label{S:numerics}

Some numerical results in 2-dimensions were presented in \cite{Terzioglu}. Here, we address the much more complicated 3-dimensional case, where we develop and apply three different inversion algorithms and study their feasibility.

Our first attempt has been to implement numerically the inversion formula (\ref{Radon by cone_odd}) from Theorem \ref{Radon in terms of cone}. The results were discouraging. The reason for this failure was that (\ref{Radon by cone_odd}) requires numerical computation of some singular integrals, followed then by applying to the results a fourth order differential operator on the sphere.

Thus we had to resort to different inversion techniques, the description of which one finds below.

In all examples below, the two-layer detectors cover the unit sphere $\mathbb{S}^2$ in $\mathbb{R}^3$ and the object is located inside of this sphere and at some positive distance from it\footnote{The spherical geometry of the detector and of most of the phantoms we consider does not constitute any inverse crime. This particular geometry is used to reduce immense computations of the synthetic \textbf{forward} data, which run for a long time even on multi-core machines. The inversion algorithms are not aware of the symmetry of the detectors and/or phantoms.}. The algorithm given in \cite{Persson} is used to generate the triangular mesh on $\mathbb{S}^2$. The forward simulations of Compton camera data were done numerically rather than analytically and thus involved errors, which is in fact better for checking the validity and stability of the reconstruction algorithms.

\subsection{Method 1: Reconstruction using spherical harmonics expansions}\label{SS:harmonics}
In this section, we derive a series formula that recovers the Radon data from cone data. Let us introduce the function
$$
G(u,\beta) := \int\limits_0^\pi Cf(u,\beta,\psi)\sin\psi d\psi.
$$
For each fixed detector location $u \in \mathbb{R}^n$, we can expand the function $G(u,\beta)$ of $\beta \in S^{n-1}$ into spherical harmonics $Y_l^m$:
 \begin{align}\label{SpHarExp}
G(u,\beta)= \sum_{l=0}^\infty \sum_{m = 1}^{N(n,l)} g_l^m(u) Y_l^m(\beta),
\end{align}
where
$$
g_l^m(u) = \int\limits_{S^{n-1}} G(u,\beta)\overline{{Y_l^m(\beta)}} d\beta
$$
and
$$
N(n,l) = (n+2l-2)\frac{(n+l-3)!}{l!(n-2)!}
$$
(see e.g. \cite{Muller,Stein,Basko}). Using \eqref{Radon by cone_odd}, one obtains the following series inversion formula:
\begin{theorem} For any $ u \in \mathbb{R}^n$ and $\omega \in S^{n-1}$,
\begin{align}\label{RadonbySpHar}
 Rf(\omega, \omega \cdot u) = \frac{\Gamma(\frac{n+1}{2})}{\pi^{n/2}}g_0^1(u) -\frac{2\pi^{-n/2}}{\Gamma(\frac{n}{2})}\sum_{l=1}^\infty d_l q_{n,l} \sum_{m = 1}^{N(n,l)} g_l^m(u) Y_l^m(\omega),
 \end{align}
 where
 \begin{equation}\label{E:q_nl}
 q_{n,l} = 4^{-(n+1)/2}\prod_{k=0}^{(n-1)/2}\left[l(l+n-2)+(2k-1)(n-1-2k)\right]
 \end{equation}
 and
 \begin{equation}\label{E:d_l}
 d_l = |S^{n-2}| \int\limits_{-1}^1 \log \frac{1}{|t|} p_l(t)(1-t^2)^{(n-3)/2}dt,
  \end{equation}
  with $p_l$ being the $l$-th degree Legendre polynomial (see \cite{Szego} or \cite[Formulas (A.7.2), (A.7.3), and (A.6.13)]{Rubin}).
\end{theorem}
\begin{proof}
Plugging \eqref{SpHarExp} into the second term in the right hand side of \eqref{Radon by cone_odd}, we obtain
  \begin{equation}
  \begin{array}{l}
 Rf(\omega, \omega \cdot u) =  \frac{\Gamma(\frac{n+1}{2})}{2\pi^{(n+1)/2}}\int\limits_{S^{n-1}} G(u,\beta) d\beta\\
 -\frac{2P_{(n+1)/2}(\Delta_S)}{\pi^{n/2}\Gamma(\frac{n}{2})}\sum_{l=0}^\infty \sum_{m = 1}^{N(n,l)} g_l^m(u) \int\limits_{S^{n-1}} \log{\frac{1}{|\omega\cdot \beta|}}Y_l^m(\beta)d\beta.
 \end{array}
 \end{equation}
We note that $\int\limits_{S^{n-1}} G(u,\beta) d\beta = 2\sqrt{\pi} g_0^1(u)$.
Then Funk-Hecke formula (see e.g. \cite{Muller}) implies that
$$\int\limits_{S^{n-1}} \log\frac{1}{|\omega \cdot \beta|}Y_l^m(\beta)d\beta=d_l Y_l^m(\omega),$$
where $d_l$ is as in (\ref{E:d_l}). Also, since $\Delta_SY_l=-l(l+n-2)Y_l$,  $l=0,1,2,...$, we have $P_{(n+1)/2}(\Delta_S)Y_l= q_{n,l} Y_l$, where $q_{n,l}$ is given in (\ref{E:q_nl}). Hence, we get the result.
\end{proof}
In particular, for $n=3$, we get
\begin{corollary} For any $ u \in \mathbb{R}^3$ and $\omega \in S^2$,
\begin{align}\label{RadonbySpHar}
 Rf(\omega, \omega \cdot u) =    \pi^{-3/2}g_0^1(u) - \frac{1}{4\pi^2}\sum_{l=1}^\infty d_l q_l \sum_{m = 1}^{2l+1} g_l^m(u) Y_l^m(\omega),
 \end{align}
 where $q_l = (l-1)l(l+1)(l+2)$ and $d_l = 2\pi \int\limits_{-1}^1 \log \frac{1}{|t|} p_l(t) dt$ with $p_l$ being the $l$-th degree Legendre polynomial.
 \end{corollary}
\begin{remark}
\normalfont The coefficients $q_l$ in \eqref{RadonbySpHar} are fourth order polynomials in $l$ and account for fourth order differentiation. Thus, it is expected to face an instability issue in the numerical implementation of \eqref{RadonbySpHar} when considering high degree spherical harmonics.
\end{remark}
In our numerical tests, the phantom was the characteristic function of the $3D$ ball of radius $0.5$ centered at the origin, while the Compton detectors covered the concentric unit sphere. The reason for considering a radial phantom is that its Radon transform can easily be computed analytically. On the other hand, the Compton data was simulated numerically and then used to numerically reconstruct the Radon data via \eqref{RadonbySpHar}. The results can then be compared with the exact (analytically computed) Radon transforms\footnote{Tests on non-radial phantoms have lead to similar results.}.

Figure \ref{fig:RadonbySpHar} shows the comparison of the analytically computed Radon transform of the phantom (shown in red) with its reconstructions, using (\ref{RadonbySpHar}). The results are illustrated for the direction $\omega = [-0.2342,   -0.1844,   -0.9545]$. In obtaining the Radon data  $Rf(\omega, s)$ for uniformly sampled $s \in [-1,1]$ from $Rf(\omega, u \cdot \omega)$,  we used MATLAB\textsuperscript{\textregistered} toolbox \verb"cftool" with spline fitting having a smoothing parameter 0.99. The cone data is numerically simulated for 1806 detector points on the sphere and 90 opening angles $\psi$. For the cone axis direction vectors, we used varying discretization of the sphere corresponding to 1806, 7446, and 30054 points. We have considered spherical harmonics up to degree $l=L=30$ in the expansion \eqref{SpHarExp}. In order to reduce the effect of instability, we only used $l=L_t$ in the computation of the Radon transform via \eqref{RadonbySpHar} equal to $18$.
\begin{figure}[H]
\begin{center}
        \begin{subfigure}[b]{0.35\textwidth}
                \includegraphics[width=\textwidth]{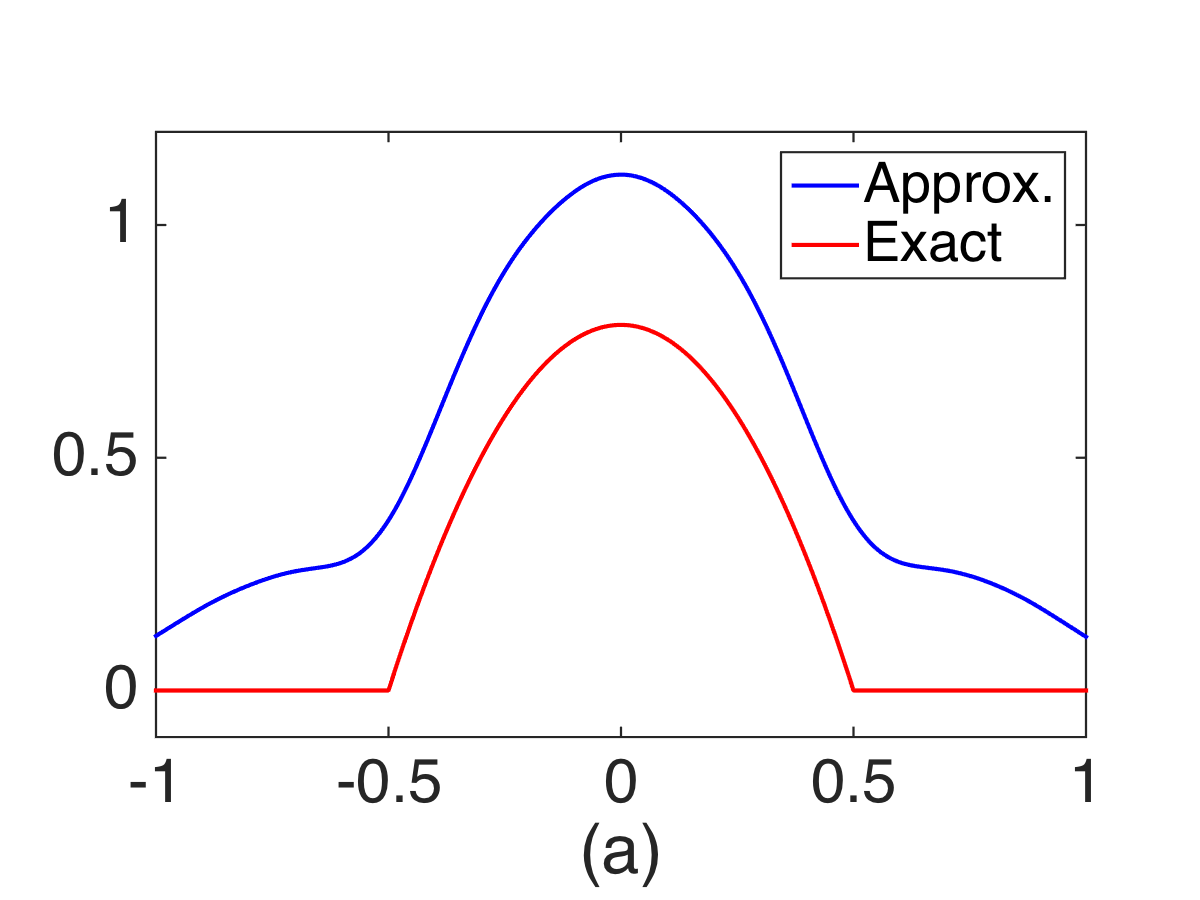}
        \end{subfigure}
        \hspace{-1.5em}
        \begin{subfigure}[b]{0.35\textwidth}
                \includegraphics[width=\textwidth]{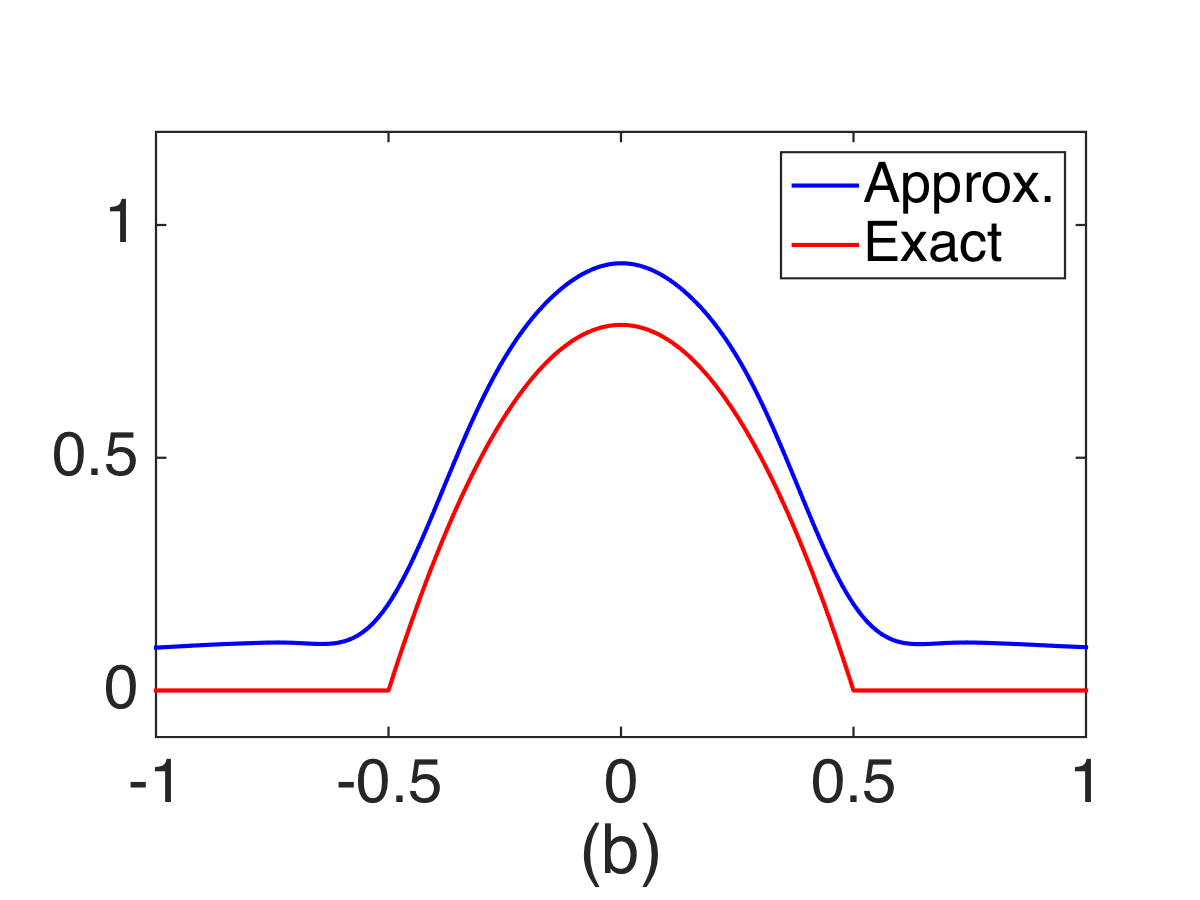}
        \end{subfigure}
        \hspace{-1.5em}
                \begin{subfigure}[b]{0.35\textwidth}
                \includegraphics[width=\textwidth]{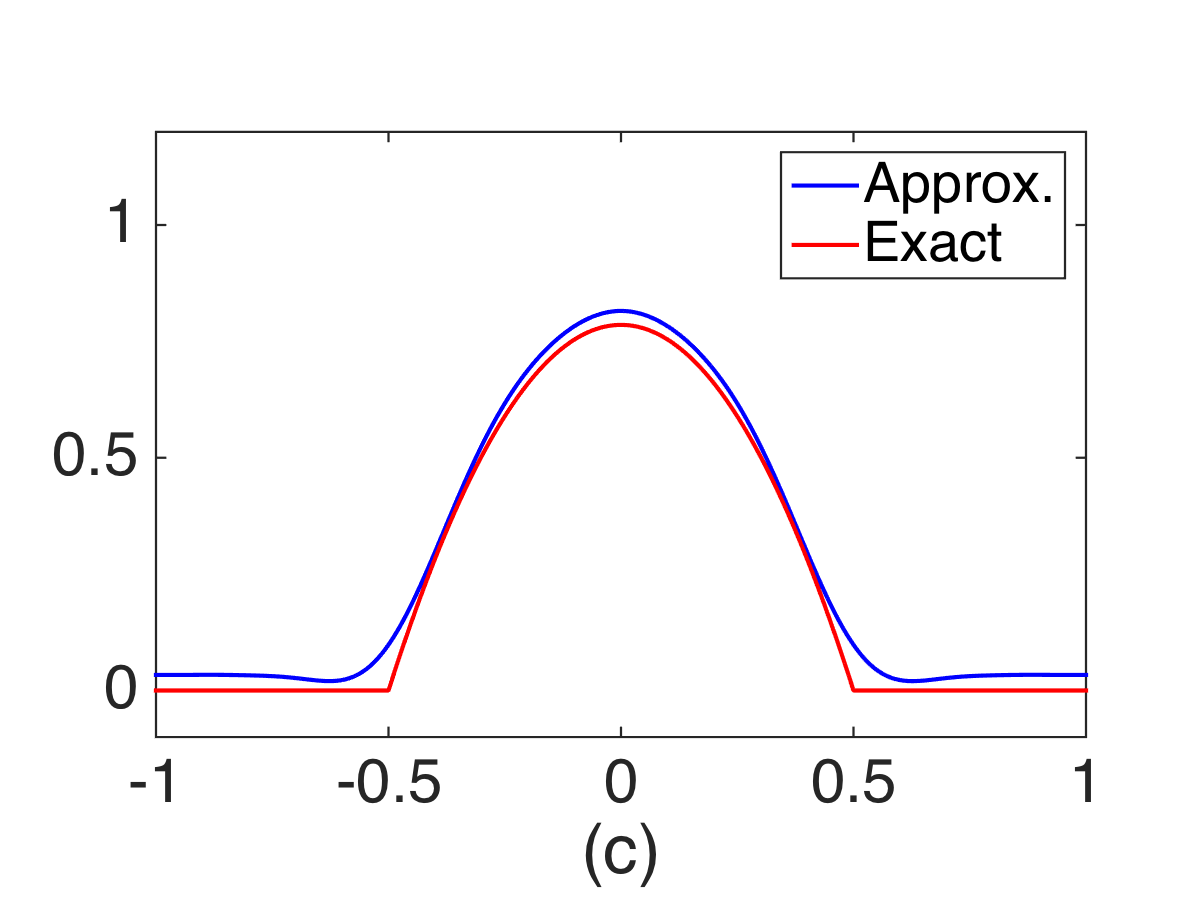}
        \end{subfigure}
\caption{The analytically computed Radon transform of the phantom (shown in red) vs. its reconstruction from the Compton data using \eqref{RadonbySpHar}. The reconstructions shown correspond to three different mesh sizes: the number of points on the sphere being 1806, 7446, and 30054, from left to right.}
\label{fig:RadonbySpHar}
\end{center}
\end{figure}

\subsection{Method 2: Reconstruction by direct implementation of Theorem \ref{C2RbyFunk}}\label{SS:Funk}

As we have mentioned before, the direct numerical implementation of the formula \eqref{Radon by cone_odd} in 3D required the application of the fourth order differential operator $\Delta_S(\Delta_S+2)$ on the sphere to the result of numerical implementation of a singular integral. The authors could not make it work well.
The advantage of using \eqref{inversion_funk} is that one needs to apply two second order operators acting in different variables and with a smoothing operator sandwiched in between. This makes such a calculation more feasible.

In our numerical implementations, we used the algorithm for the discrete Laplace-Beltrami operator given in \cite{Belkin}, which comprises heat equation based smoothing used to create a point-wise convergent approximation for the Laplace-Beltrami operator on a surface. For a function $f$ given at the set $V$ of vertices of a mesh $K$ on the 2-sphere, it is computed, for any $v \in V$, as follows:
\begin{align}\label{DLB}
\Delta_K^hf(v)=\frac{1}{4\pi h^2} \sum_{t \in K} \frac{Area(t)}{\#t} \sum_{p \in V(t)} e^{-\frac{\|p-v\|^2}{4h}}(f(p)-f(v)).
\end{align}
Here, for any face $t \in K$, the number of vertices in $t$ is denoted by $\# t$, and $V(t)$ is the set of vertices of $t$. The parameter $h$ is a positive quantity (akin to the time in the heat equation), which intuitively corresponds to the size of the neighborhood considered at each point. The authors of \cite{Belkin} suggest that $h$ can be taken to be a function of $v$, which allows the algorithm to adapt to the local mesh size.

In our experiments, we used the adaptive parameter $h(v) = 0.0156\times$(the average edge length at $v$). We used the same phantom as in the previous section. Figure \ref{fig:RadonbyDLB} shows the comparison of the analytically computed Radon transform of the phantom (shown in red) with its reconstructions. The results are illustrated for the direction $\omega = [-0.2363,   -0.2484,   -0.9394]$. In obtaining the Radon data for uniformly sampled $s \in [-1,1]$ from $Rf(\omega, u \cdot \omega)$,  we used the same MATLAB\textsuperscript{\textregistered} toolbox \verb"cftool" with spline fitting having a smoothing parameter 0.995. The cone data is numerically simulated for 1806 detector points on the sphere and 90 opening angles $\psi$. For the cone axis direction vectors, we used varying discretization of the sphere corresponding to 1806, 7446, and 30054 points.
\begin{figure}[H]
\begin{center}
        \begin{subfigure}[b]{0.35\textwidth}
                \includegraphics[width=\textwidth]{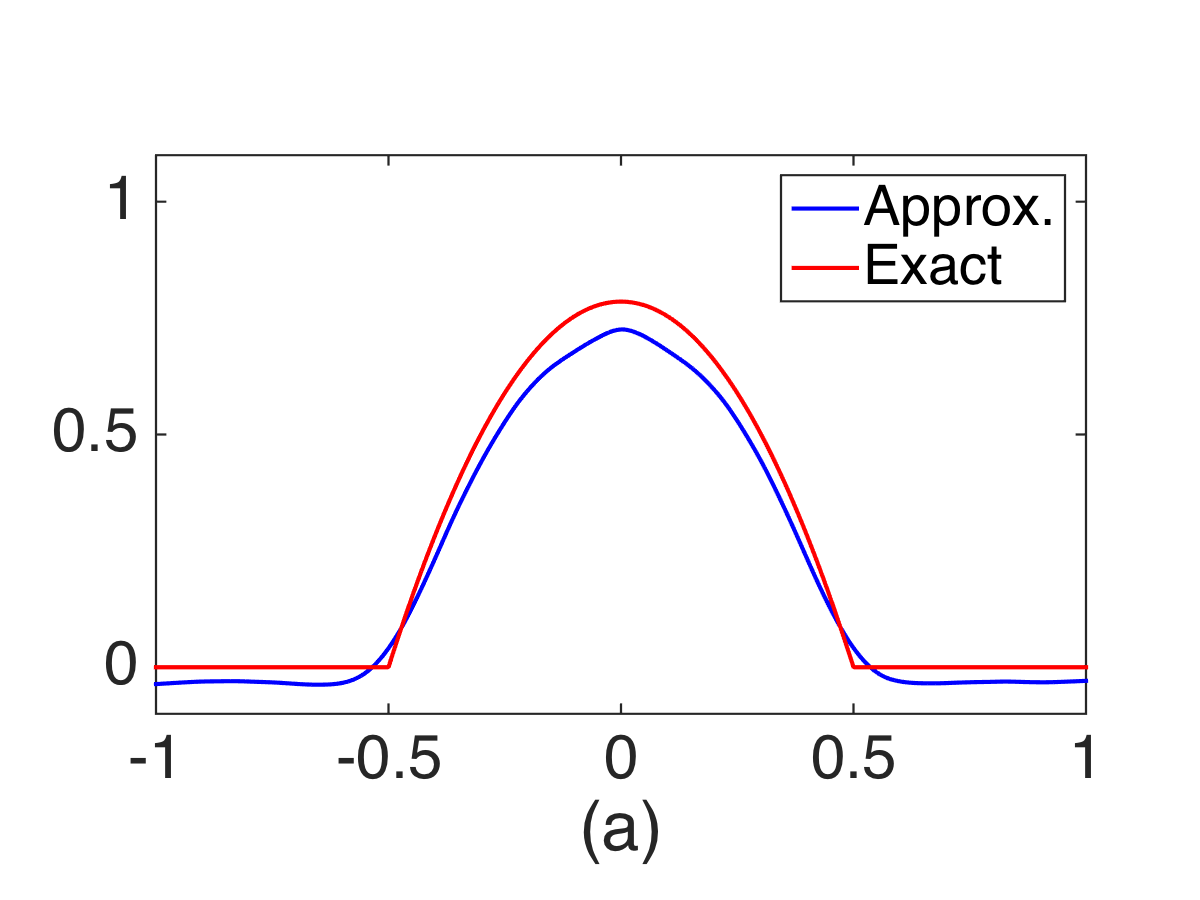}
        \end{subfigure}
        \hspace{-1.5em}
        \begin{subfigure}[b]{0.35\textwidth}
                \includegraphics[width=\textwidth]{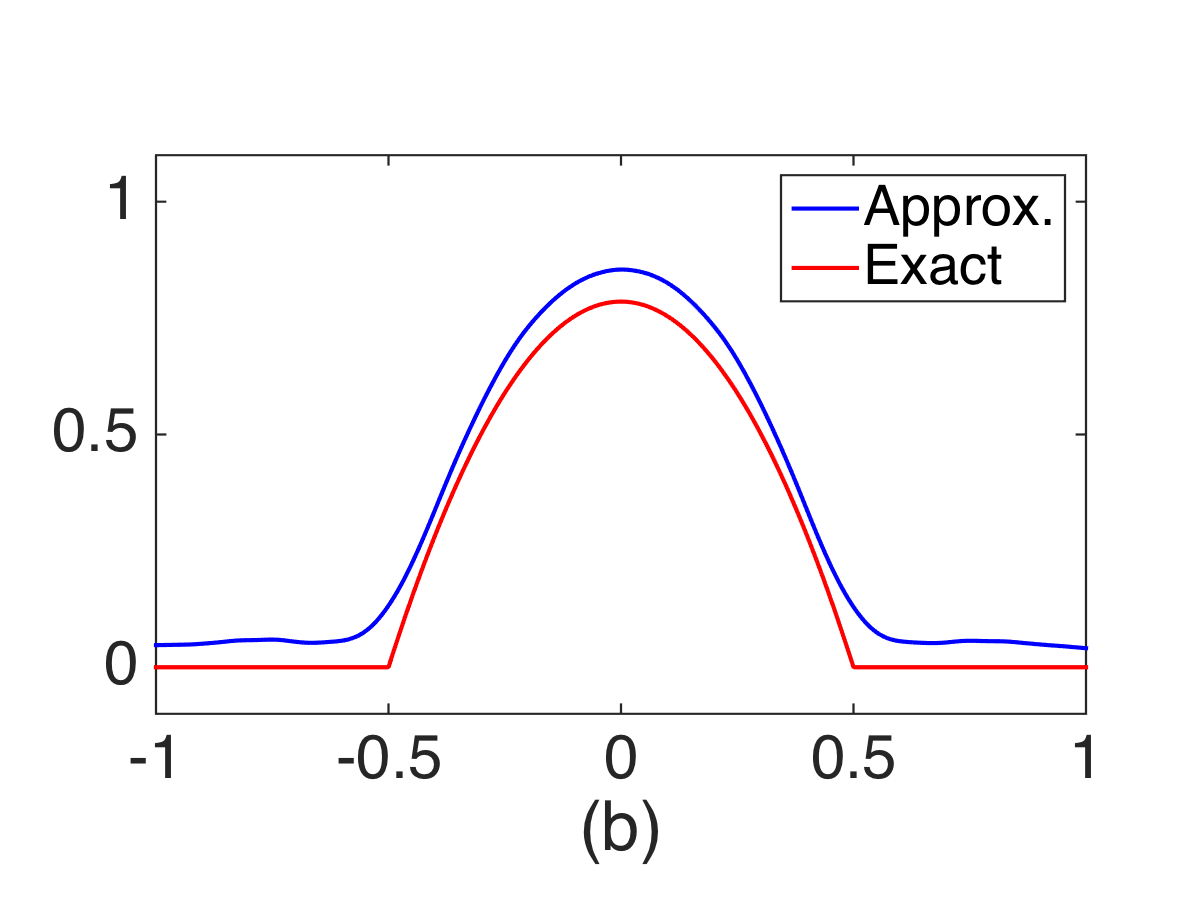}
        \end{subfigure}
        \hspace{-1.5em}
                \begin{subfigure}[b]{0.35\textwidth}
                \includegraphics[width=\textwidth]{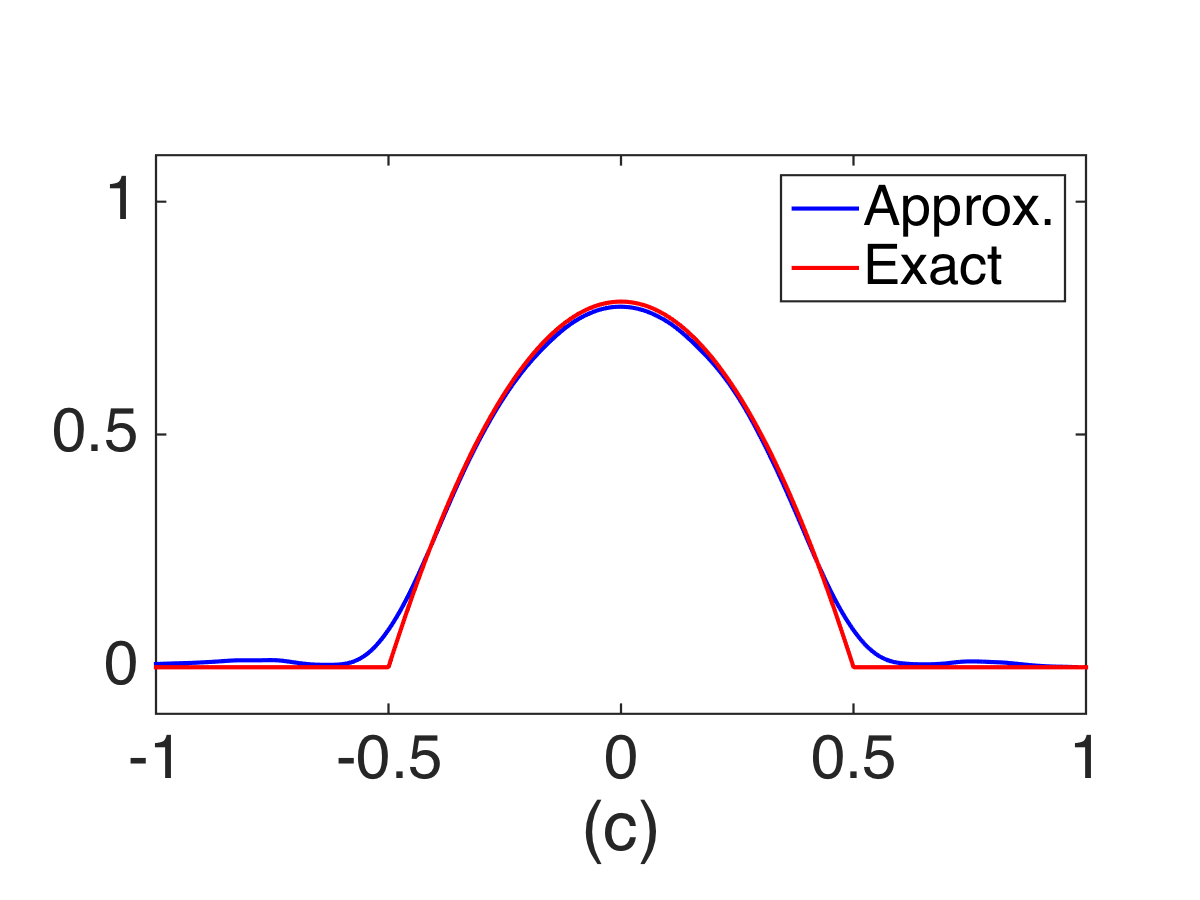}
        \end{subfigure}
\caption{The Radon transform of the phantom recovered using \eqref{inversion_funk} and \eqref{DLB}. The reconstructions shown corresponds to three different mesh sizes: the number of points on the sphere being 1806, 7446, and 30054, from left to right.}
\label{fig:RadonbyDLB}
\end{center}
\end{figure}
\subsection{Method 3: Reconstruction via a mollified inversion of the Cosine Transform}\label{SS:cosine}

The formula \eqref{int_rel2} shows that availability of any cosine transform inversion would also lead to an inversion of the cone transform, and such approximate and exact inversions of $\C$ indeed exist \cite{Louis, Riplinger, Rubin}. We apply here the method of approximate inverse developed in \cite{Louis&Maass, Louis, Riplinger}, which is an incarnation of a general approach to solving inverse problems numerically. Namely, for a given data $h$, the aim is to find $g$ satisfying $\C g=h$. If we find a `Green's function' $\psi$ such that $\C \psi=\delta$, then the spherical convolution $h \ast \psi$ of $h$ and $\psi$ solves the equation $\C g=h$. Now, if one picks a `mollifier' (an approximation to the $\delta$-function) $\delta_\gamma$ and ``approximate Greens function'' $\psi_\gamma$, such that $\C \psi_\gamma=\delta_\gamma$, then one finds the approximate solution $g_\gamma=g \ast \delta_\gamma$.

In our numerical tests, we used the reconstruction kernel $\psi_\gamma$ that was analytically computed in \cite{Riplinger} for a special class of mollifiers (see \cite[(4.1) and (4.10)]{Riplinger}). We used the same phantom as in the previous sections. Figure \ref{fig:RadonbyAI} shows the comparison of the analytically computed Radon transform of the phantom (shown in red) with its reconstructions, using (\ref{RadonbySpHar}). The results are illustrated for the direction $\omega = [-0.2342,   -0.1844,   -0.9545]$. The  MATLAB\textsuperscript{\textregistered} toolbox \verb"cftool" was used with the smoothing parameter 0.99999999. The cone data is numerically simulated for 1806 detector points on the sphere and 90 opening angles $\psi$. For the cone axis direction vectors, we used varying discretization of the sphere corresponding to 1806, 7446, and 30054 points.
\begin{figure}[H]
\begin{center}
        \begin{subfigure}[b]{0.35\textwidth}
                \includegraphics[width=\textwidth]{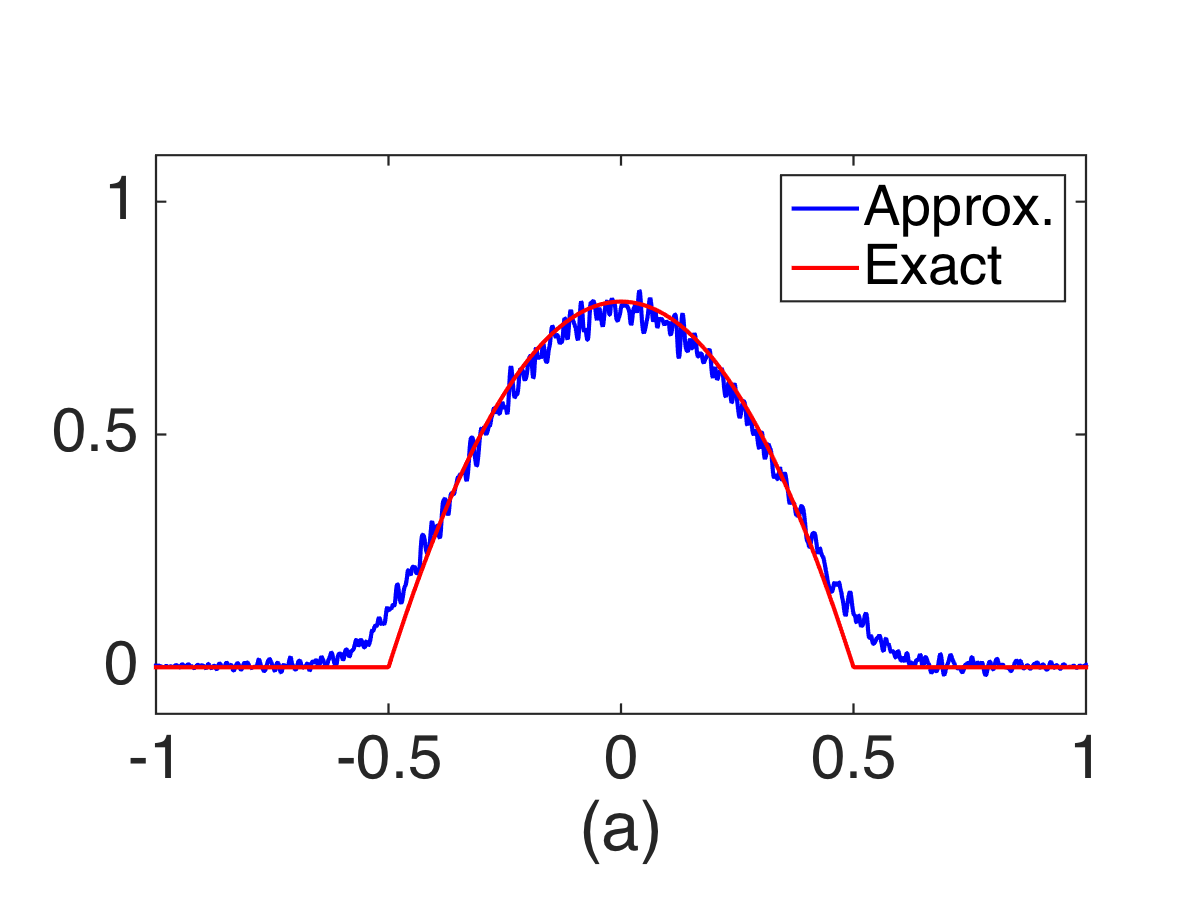}
        \end{subfigure}
        \hspace{-1.5em}
        \begin{subfigure}[b]{0.35\textwidth}
                \includegraphics[width=\textwidth]{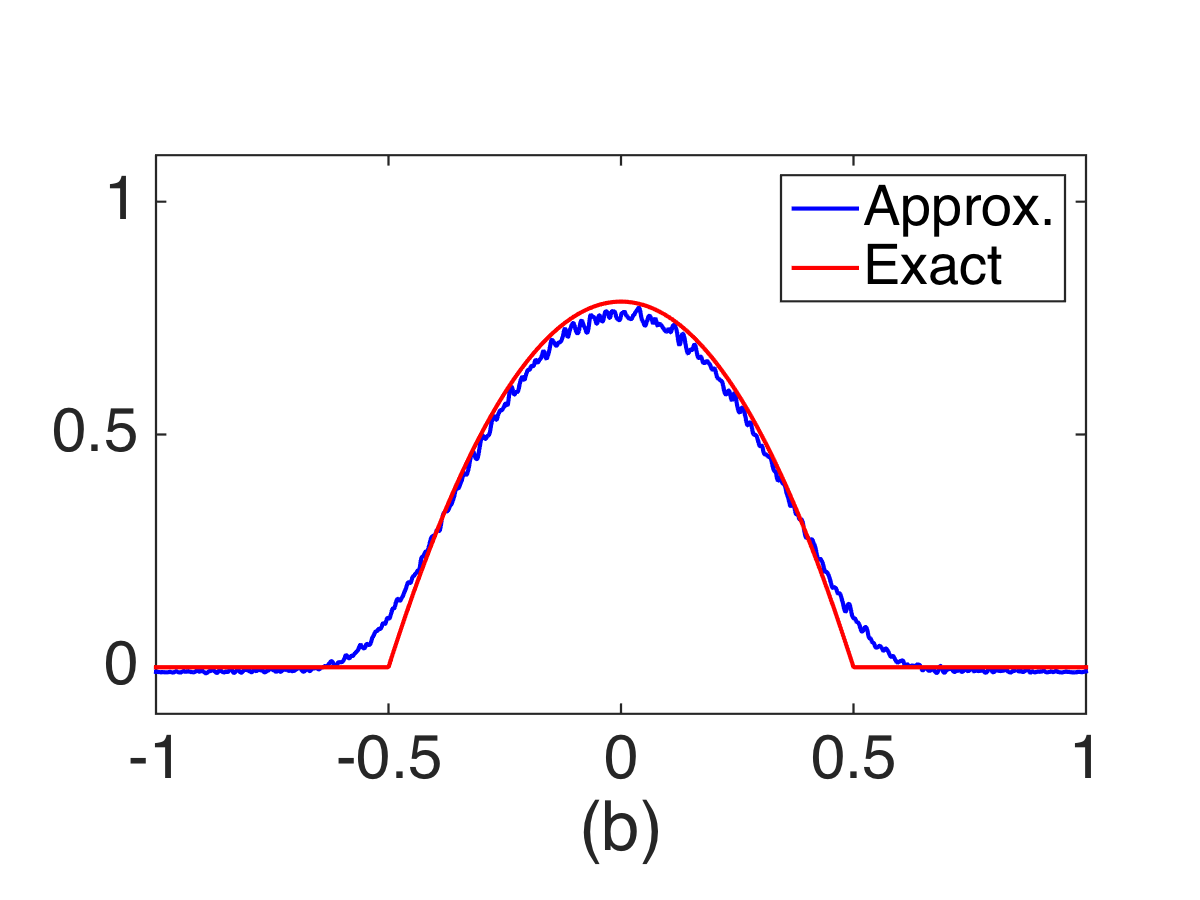}
        \end{subfigure}
        \hspace{-1.5em}
                \begin{subfigure}[b]{0.35\textwidth}
                \includegraphics[width=\textwidth]{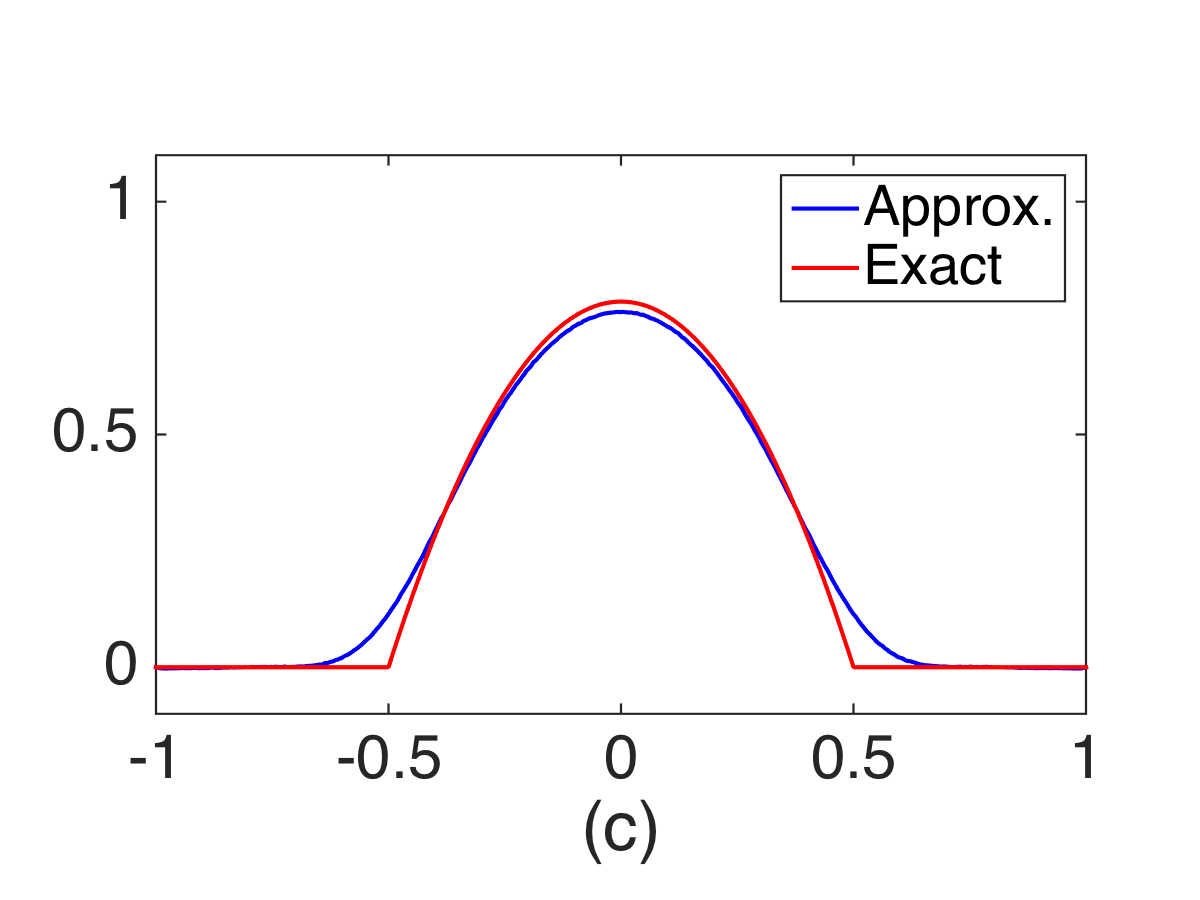}
        \end{subfigure}
 \caption{The Radon transform of the phantom using the method of mollified inverse for the cosine transform. The reconstructions shown corresponds to three different mesh sizes: the number of points on the sphere being 1806, 7446, and 30054, from left to right.}
\label{fig:RadonbyAI}
\end{center}
\end{figure}
One notices insufficient resolution of singularity, which is due to the insufficiently fine approximation of $\delta$-function by $\delta_\gamma$ chosen in \cite{Louis, Riplinger}.

\subsection{Comparison of the three methods}\label{SS:compare}
While above we only addressed reconstructing the Radon transform of the function in question, here we show how the three methods perform after taking the final step of inverting the Radon transform and reconstructing the characteristic function of the ball.

The inversion of Radon transform from the reconstructed values $Rf(\omega,s)$ was done according to the formula (\ref{inverse_radon}) with $\alpha=0$. We used 128 values of $s$ and 480 directions $\omega$ in methods 1 and 3, and 1806 directions in method 2. We used the filtered backprojection formula with the filter given in \cite{Marr}. The normalized $L^2$ and $H^1$ errors for the Radon transforms obtained in each of the methods are summarized in Table \ref{tab:RadonError}. The reason for considering the $H^1$-error is the fact that $H^1$-norm control of the $3D$ Radon transform data $Rf$ corresponds to the $L^2$-norm control of the tomogram $f$ \cite{Natt_old}.
\begin{table}[H]
\begin{center}
\begin{tabular}{ | c | c | c |}
\hline
  Method & $L^2$ Error & $H^1$ Error \\
  \hline
  1 & 0.0986 & 0.3231\\
  \hline
  2 & 0.1046 & 0.3767 \\
  \hline
  3 & 0.0896 & 0.3660 \\
  \hline
\end{tabular}
\caption{The normalized $L^2$ and $H^1$ errors for the Radon data for each of the three methods.}
\label{tab:RadonError}
\end{center}
\end{table}
\vspace{0.1em}

Figure \ref{fig:Methods} shows the three cross-sections of the spherical phantom and of its reconstructions from the Radon data obtained via the three methods above. The finest mesh on the sphere (30054 points) was used.
\begin{figure}[H]
\begin{center}
        \begin{subfigure}[b]{0.49\textwidth}
                \includegraphics[width=\textwidth]{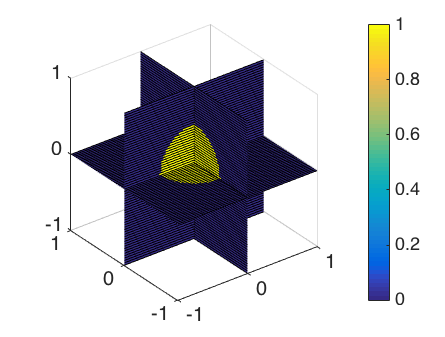}
                \caption{}
        \end{subfigure}
        \begin{subfigure}[b]{0.49\textwidth}
                \includegraphics[width=\textwidth]{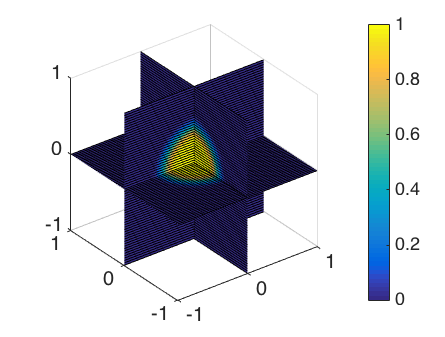}
                \caption{}
        \end{subfigure}

                \begin{subfigure}[b]{0.49\textwidth}
                \includegraphics[width=\textwidth]{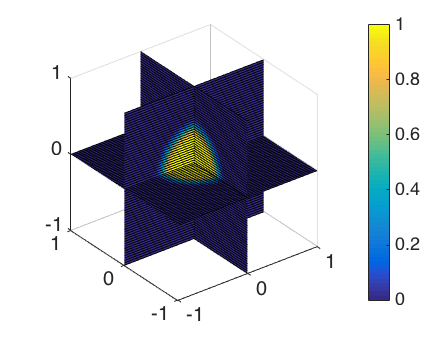}
                \caption{}
        \end{subfigure}
        \begin{subfigure}[b]{0.49\textwidth}
                \includegraphics[width=\textwidth]{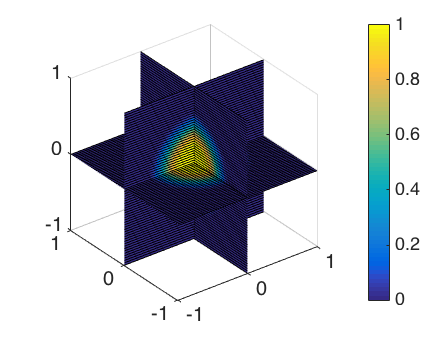}
                \caption{}
        \end{subfigure}
        \caption{Comparison of the three reconstruction methods. The cross-sections by the coordinate planes are shown. (a) The phantom is the characteristic function of 3d ball having radius 0.5 and center at the origin. (b) Reconstruction via Method 1. (c) Reconstruction via Method 2. (d) Reconstruction via Method 3.}
        \label{fig:Methods}
\end{center}
\end{figure}
Figure \ref{fig:Profiles} shows $x$-profiles of the central cross-sections of the spherical phantom and of its reconstructions shown in Figure \ref{fig:Methods}.
\begin{figure}[H]
\begin{center}
        \begin{subfigure}[b]{0.35\textwidth}
                \includegraphics[width=\textwidth]{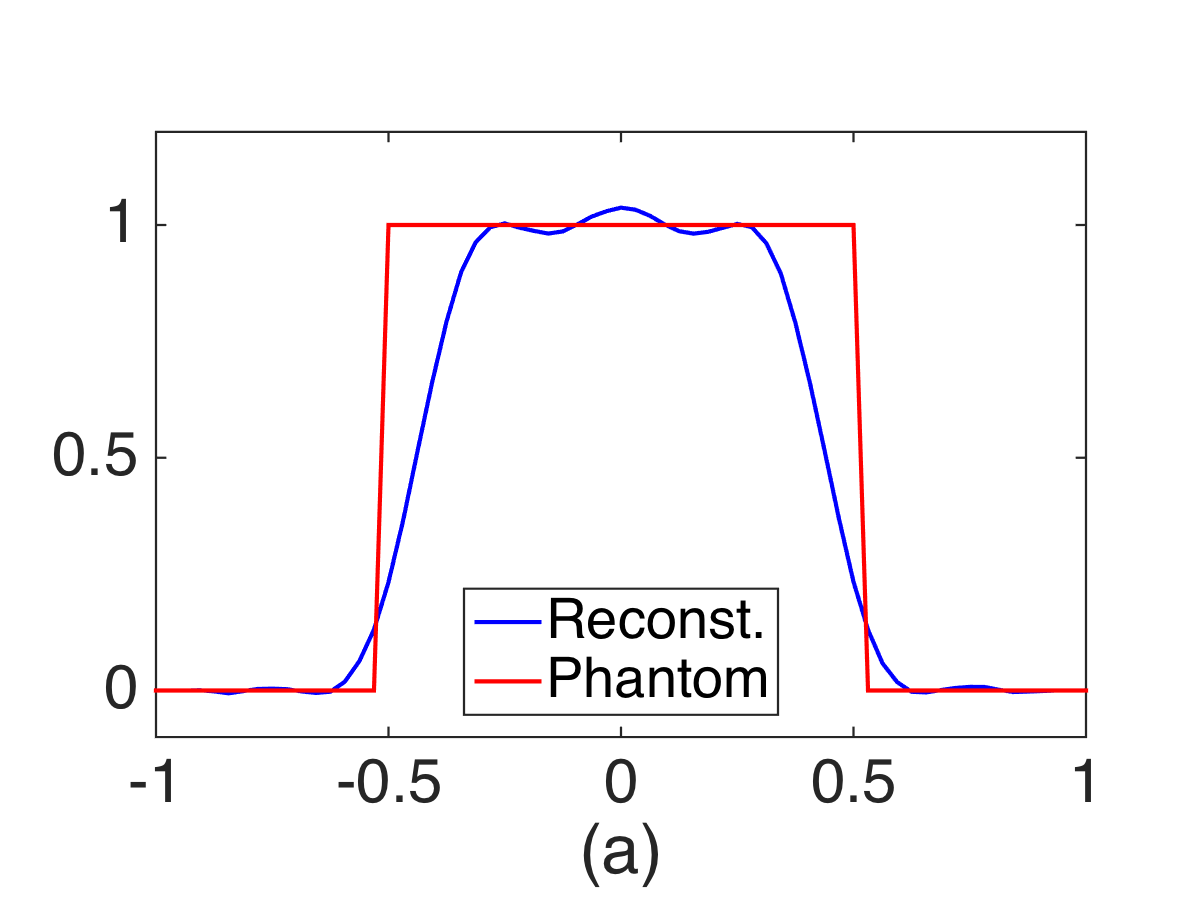}
        \end{subfigure}
        \hspace{-1.5em}
        \begin{subfigure}[b]{0.35\textwidth}
                \includegraphics[width=\textwidth]{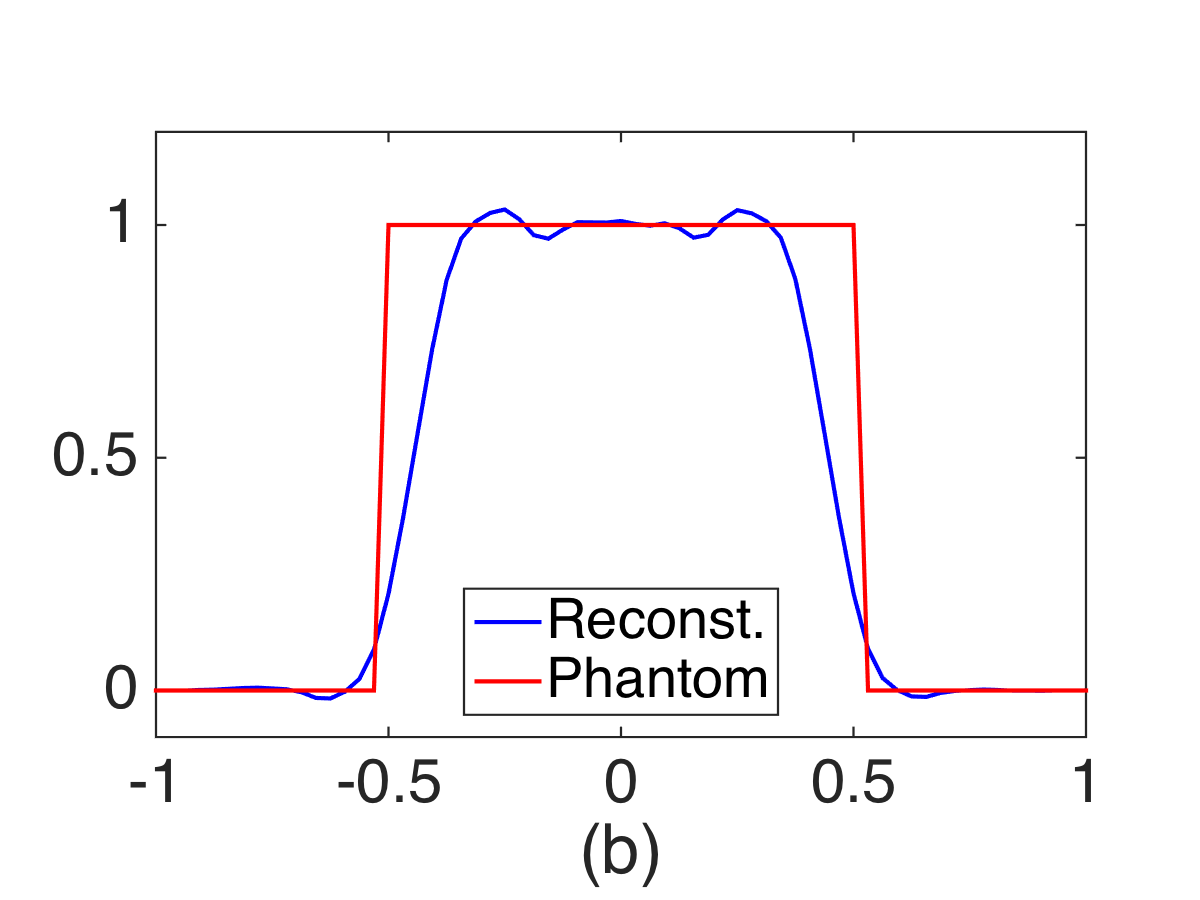}
       \end{subfigure}
       \hspace{-1.5em}
        \begin{subfigure}[b]{0.35\textwidth}
        \includegraphics[width=\textwidth]{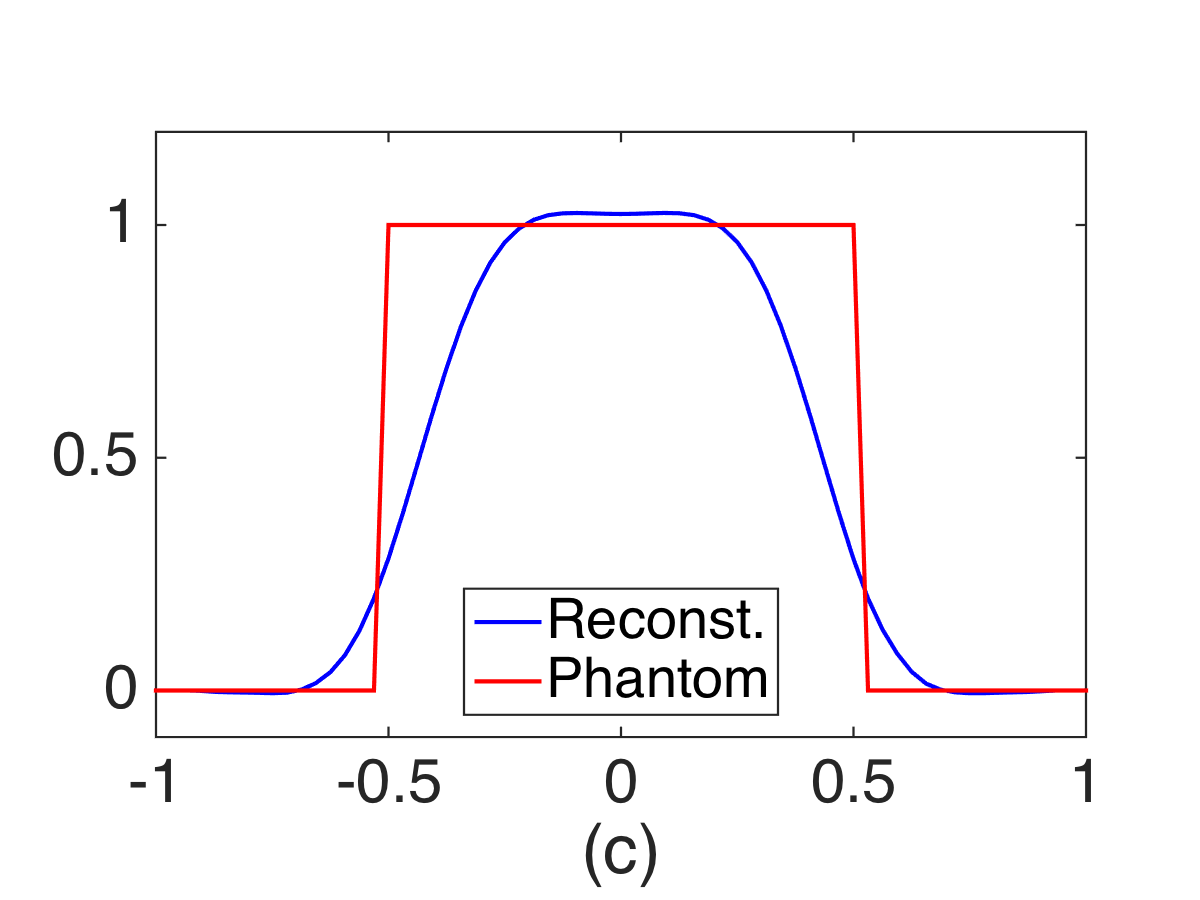}
        \end{subfigure}
        \caption{x-profiles of the phantom and the reconstructions in Figure \ref{fig:Methods}, (a) method 1, (b) method 2 and (c) method 3.}
        \label{fig:Profiles}
\end{center}
\end{figure}

 It is important to note that in all of the methods, there are parameters that can still be optimized, namely $L$ and $L_t$ in Method 1, $h$ in Method 2, and $\gamma$ and $\nu$ in Method 3 (see \cite{Riplinger}).

 We have also tested the reaction of our algorithms to random noise. The $20\%$ Gaussian white noise added to the cone data for Methods 1 and 2. For Method 2, we added $10\%$ noise to the cone data. Figure \ref{fig:MethodsNoisy} shows the three cross-sections of the spherical phantom and of its reconstructions from the Radon data obtained via the three methods above. The finest mesh on the sphere (30054 points) was used.
 \begin{figure}[H]
\begin{center}
        \begin{subfigure}[b]{0.49\textwidth}
                \includegraphics[width=\textwidth]{phantom3D.png}
                \caption{}
        \end{subfigure}
        \begin{subfigure}[b]{0.49\textwidth}
                \includegraphics[width=\textwidth]{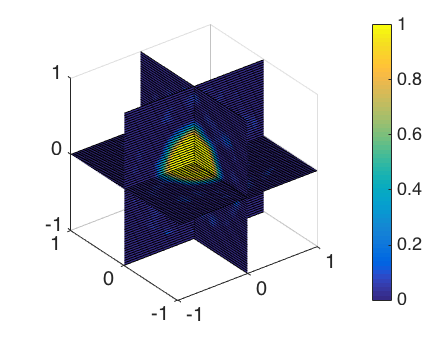}
                \caption{}
        \end{subfigure}

                \begin{subfigure}[b]{0.49\textwidth}
                \includegraphics[width=\textwidth]{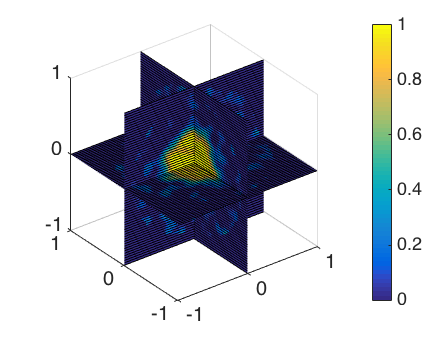}
                \caption{}
        \end{subfigure}
        \begin{subfigure}[b]{0.49\textwidth}
                \includegraphics[width=\textwidth]{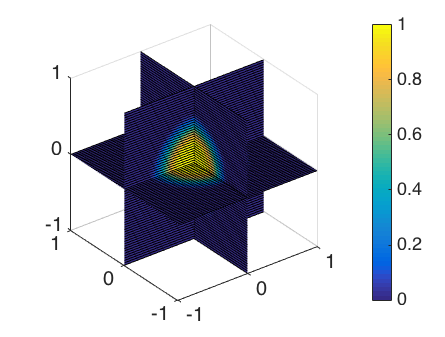}
                \caption{}
        \end{subfigure}
        \caption{Comparison of the three reconstruction methods. The cross-sections by the coordinate planes are shown. (a) The phantom is the characteristic function of 3d ball having radius 0.5 and center at the origin. (b) Reconstruction via Method 1 from data contaminated with $20\%$ Gaussian white noise. (c) Reconstruction via Method 2 from data contaminated with $10\%$ Gaussian white noise. (d) Reconstruction via Method 3 from noisy data contaminated with $20\%$ Gaussian white noise.}
        \label{fig:MethodsNoisy}
\end{center}
\end{figure}
Figure \ref{fig:ProfilesNoisy} shows $x$-profiles of the central cross-sections of the spherical phantom and of its reconstructions shown in Figure \ref{fig:MethodsNoisy}.
\begin{figure}[H]
\begin{center}
        \begin{subfigure}[b]{0.35\textwidth}
                \includegraphics[width=\textwidth]{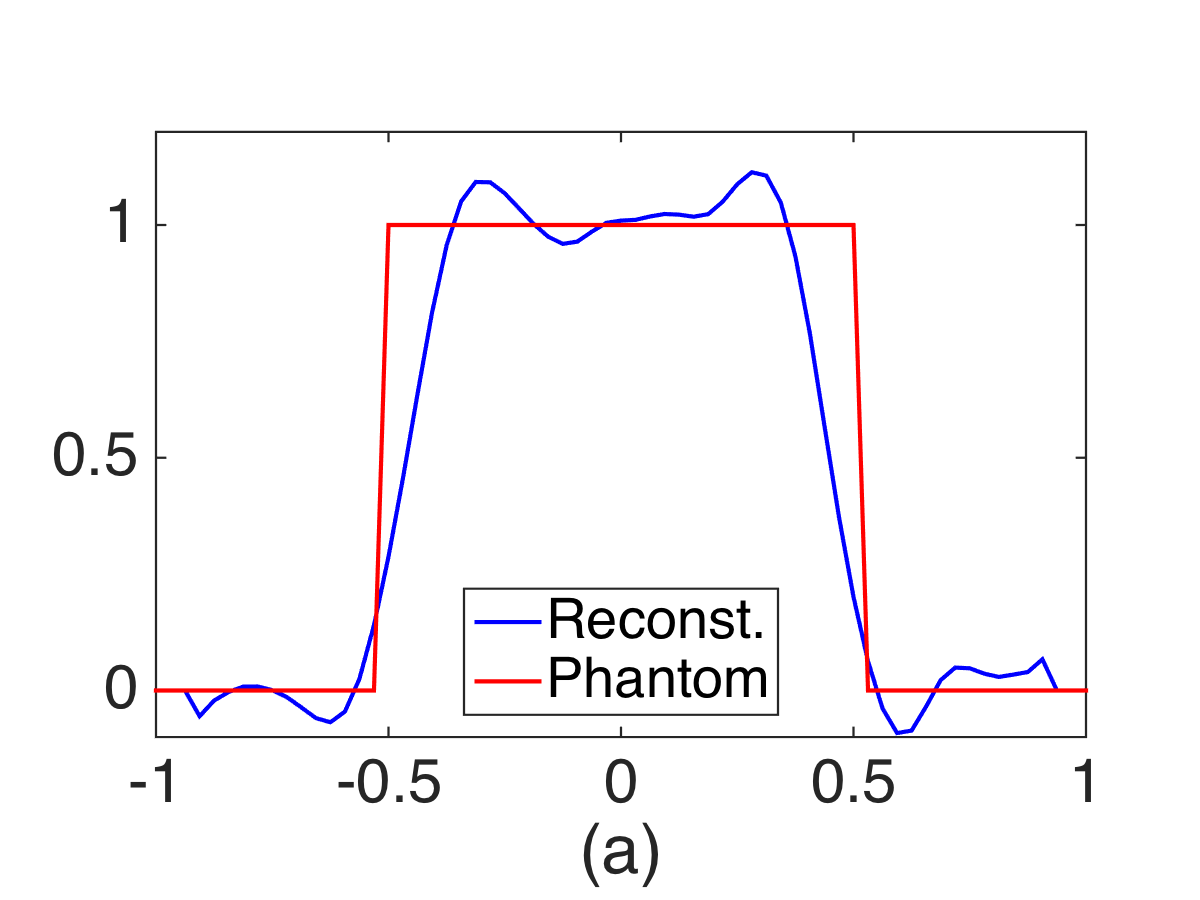}
        \end{subfigure}
        \hspace{-1.5em}
        \begin{subfigure}[b]{0.35\textwidth}
                \includegraphics[width=\textwidth]{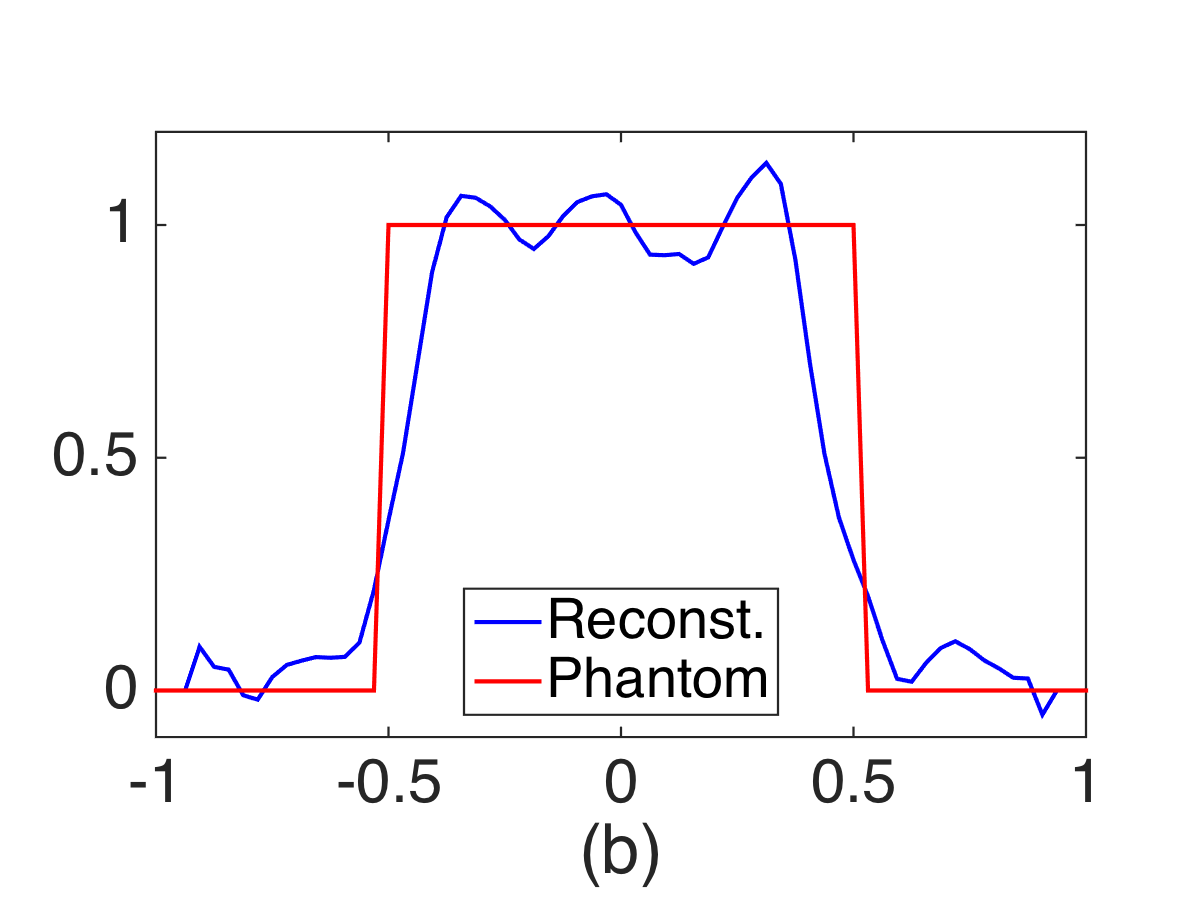}
       \end{subfigure}
       \hspace{-1.5em}
                \begin{subfigure}[b]{0.35\textwidth}
                \includegraphics[width=\textwidth]{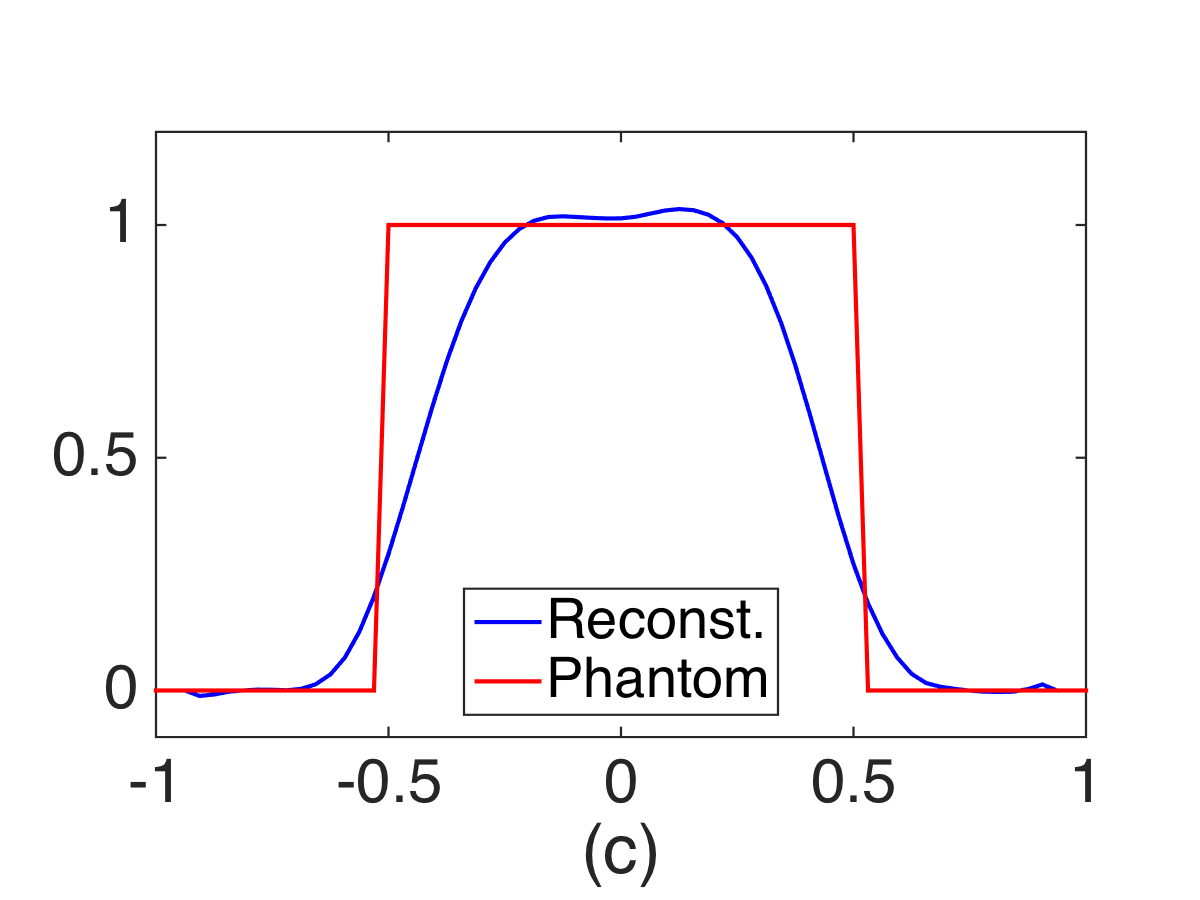}
        \end{subfigure}
        \caption{Comparison of x-profiles of central slices of phantom and the reconstructions from noisy data shown in Figure  \ref{fig:MethodsNoisy}.}
        \label{fig:ProfilesNoisy}
\end{center}
\end{figure}

\section{Conclusion and Remarks}\label{S:remarks}
\begin{enumerate}
\item It is argued that in the case of Compton camera imaging, reducing the set of cones ``visible'' from a detector (e.g., considering only the cones with a given axial direction), which was done in most previous studies, seems to be not a very good idea (especially in presence of low SNR), since this amounts to discarding the already collected data while it could be used for stabilizing the reconstruction.
\item A general ``admissibility'' criterion for the geometry of the set of detectors is formulated. Under this condition, the formulas provided allow reconstructions for an otherwise arbitrary geometry of detector arrays. If the condition is violated, the reconstructions will produce the familiar \cite{Natt_old,KuchCBMS} limited data blurring artifacts.
\item Three new analytical reconstruction techniques are developed and numerically implemented for inverting the Compton camera data. Different inversion formulas (which are all equivalent on the range of the cone transform) can behave differently with respect to errors. Thus numerical comparison of the three techniques was conducted.
\item A different spherical harmonics expansion technique for recovering the Radon data and then the function from the cone data was also developed and then tested on a similar phantom in \cite{Basko}.
\item Cone data was numerically generated for a uniform spherical phantom and used by these methods to recover the phantom. The results confirm that all three techniques work, albeit react differently to the noise added. Namely, the methods 1 and 3 produce decent images even with 20 $\%$ noise, while method 2 starts breaking down at this level, but survives with 10 $\%$ noise. This is not too surprising, taking into account that in the latter case two successive Laplace-Beltrami operators need to be applied numerically.
\item All suggested numerical techniques allow for an additional fine tuning of parameters: number of terms considered in method 1, smoothing parameter $h$ in method 2, and parameters of the mollifier in method 3.
\item The simple ball phantom was used for the following reason. The methods are two-step: recovery of the Radon data from the cone data, and then the Radon transform inversion. The ball phantom allows one to judge the effect of the cone data inversion alone, since the ball's Radon transform is known exactly. The last step, Radon transform inversion is well studied. The authors used at this step an algorithm, whose testing showed its good quality.
\item Some reviewers expressed understandable concerns about the usage of the spherical concentric phantom and detector surface. The reason why such simple geometry was chosen is that it reduces by orders of magnitude the very heavy (even on fast multi-core machines with parallel algorithms) computational cost of producing the \textbf{synthetic forward} data to test the algorithms. The inversion algorithms we described do not contain any information about the geometry of the detectors and the phantom. As soon as we achieve faster forward algorithms, we will be able to use arbitrary geometries. So far, as a partial relief we can say that the data with errors did not carry the same symmetry, and thus stability of the algorithms is somewhat reducing such concerns. Also, we provide in Fig. \ref{F:shifted} a reconstruction of the phantom, where the phantom and the detector are not concentric anymore: the center of the phantom is moved up $20\%$ of the length of the detector sphere's radius. We also had to reduce somewhat the quality of the forward data, to reduce the computation time in this case.
    \begin{figure}[H]
\begin{center}
        \begin{subfigure}[b]{0.45\textwidth}
                \includegraphics[width=\textwidth]{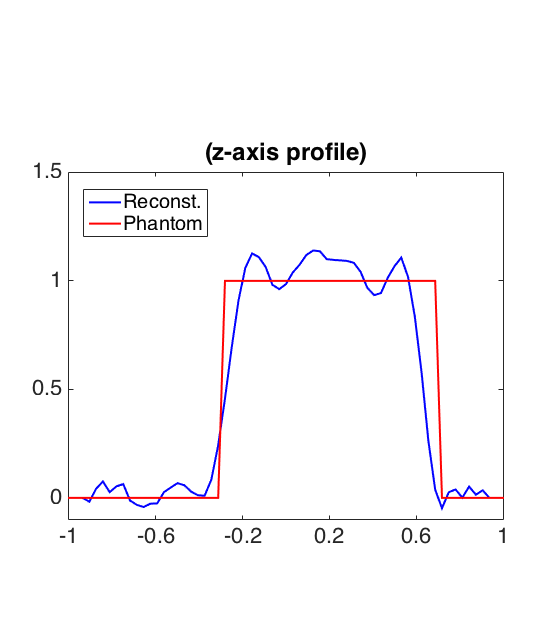}
                \caption{}
        \end{subfigure}
        \hspace{0.05em}
        \begin{subfigure}[b]{0.5\textwidth}
                \includegraphics[width=\textwidth]{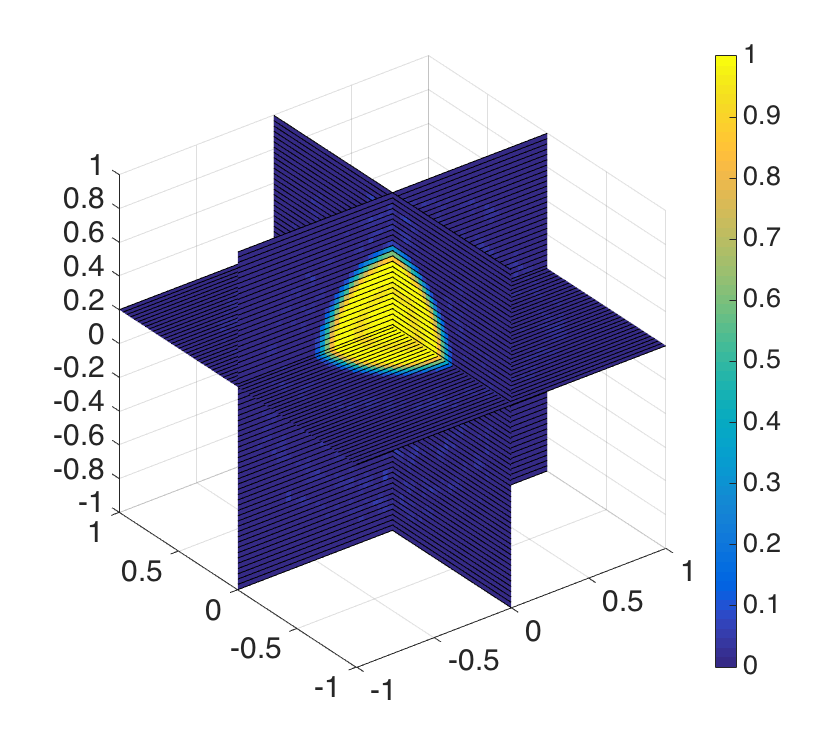}
                \caption{}
        \end{subfigure}
        \caption{Reconstruction of a non-concentric (with respect to the detectors) spherical phantom.}
        \label{F:shifted}
\end{center}
\end{figure}
\end{enumerate}
\section*{Acknowledgements}
This work was supported in part by the NSF DMS Grant \# 1211463. The authors thank NSF for this support. Thanks go to B.~Rubin for many useful comments and written materials provided. The authors are also grateful to A.~Bonito for numerous discussions and suggestions concerning numerical implementation of the algorithms and to the referees for making many valuable suggestions on improving the text. Finally, the authors thank the Numerical Analysis and Scientific Computing group in the Department of Mathematics at Texas A$\&$M University for letting the authors use their computing resources.


\end{document}